\documentclass[12pt]{article}
\usepackage[margin=1in]{geometry}
\usepackage[normalem]{ulem}

\usepackage{amsmath}
\usepackage{amsthm,physics}
\usepackage{array}
\usepackage{authblk}
\usepackage{amssymb}
\usepackage{xcolor}
\usepackage[numbers,sort&compress]{natbib}
\usepackage{hyperref}
\usepackage{mathtools}
\usepackage{xargs}[2008/03/08]
\usepackage[shortlabels]{enumitem}
\usepackage{comment}
\usepackage{subcaption}
\usepackage{mathrsfs}

\theoremstyle{plain}
\newtheorem{theorem}{Theorem}[section]
\newtheorem{lemma}[theorem]{Lemma}
\newtheorem{proposition}[theorem]{Proposition}
\newtheorem{corollary}[theorem]{Corollary}

\theoremstyle{definition}
\newtheorem{definition}[theorem]{Definition}

\theoremstyle{remark}

\newtheorem{remark}[theorem]{Remark}
\newtheorem{assumption}{Assumption}

\newtheorem{claim}{Claim}
\numberwithin{claim}{theorem}
\def\claimproof#1{\begin{proof}[Proof of Claim \ref{#1}] $\triangleright$}
\def\endclaimproof{\let\oldqed\qedsymbol\renewcommand\qedsymbol{$\triangleleft$}\end{proof}\renewcommand\qedsymbol{\oldqed}}

\numberwithin{equation}{section}

\usepackage{dsfont}
\newcommand{\indicator}[1]{\mathds{1}_{#1}}
\newcommand{\dist}{\operatorname{dist}}
\renewcommand{\cp}{\mathrm{cp}}
\newcommand{\set}[1]{\left\{#1\right\}}
\newcommand{\floor}[1]{\left\lfloor#1\right\rfloor}
\renewcommand{\max}{\mathrm{max}}

\global\long\def\per{\mathrm{per}}
\global\long\def\eps{\epsilon}
\global\long\def\E{\mathbb{E}}
\global\long\def\Z{\mathbb{Z}}
\global\long\def\P{\mathbb{P}}

\def\R{\mathbb{R}}
\def\N{\mathbb{N}}
\DeclareMathOperator\diam{diam}
\DeclareMathOperator\CE{CE}

\begin{document}

\title{Liquid-vapor transition in a model of a continuum particle system with finite-range modified Kac pair potential}

\author[1]{Qidong He}
\author[2]{Ian Jauslin}
\author[3]{Joel Lebowitz}
\author[4]{Ron Peled}
\affil[1,2,3]{Department of Mathematics, Rutgers University}
\affil[3]{Department of Physics, Rutgers University}
\affil[4]{Department of Mathematics, University of Maryland, College Park}

\date{}

\maketitle

\abstract

We prove the existence of a phase transition in dimension $d>1$ in a continuum particle system interacting with a pair potential containing a modified attractive Kac potential of range $\gamma^{-1}$, with $\gamma>0$. 
This transition is ``close'',  for small positive $\gamma$, to the one proved previously by Lebowitz and Penrose in the van der Waals limit $\gamma\downarrow0$.
It is of the type of the liquid-vapor transition observed when a fluid, like water, heated at constant pressure, boils at a given temperature. 
Previous results on phase transitions in continuum systems with stable potentials required the use of unphysical four-body interactions or special symmetries between the liquid and vapor.

The pair interaction we consider is obtained by partitioning space into cubes of volume $\gamma^{-d}$, and letting the Kac part of the pair potential be uniform in each cube and act only between adjacent cubes. 
The ``short-range'' part of the pair potential is quite general (in particular, it may or may not include a hard core), but restricted to act only between particles in the same cube.

Our setup, the ``boxed particle model'', is a special case of a general ``spin'' system, for which we establish a first-order phase transition using reflection positivity and the Dobrushin--Shlosman criterion.

\tableofcontents

\section{Introduction}

In 1998, Lebowitz, Mazel, and Presutti~\cite{lebowitz1998rigorous} wrote:
``An outstanding problem in equilibrium statistical mechanics is to derive rigorously the existence of a liquid-vapor phase transition (LVT) in a continuous system of particles interacting with any kind of reasonable potential, say Lennard--Jones or hard core plus attractive square well.''
This situation has remained largely unchanged over the past quarter century.
This is so despite the fact that the LVT is ubiquitous.
It is observed every time we boil a pot of water and is displayed prominently in textbooks as a paradigm of phase transitions in physical systems; see Figure \ref{fig:phase-diagram}.
The LVT is also observed in all computer simulations of systems with the above type of pair potentials \cite{hansen1969phase,mcdonald1972equation}.
These simulations show that such pair potentials are, in fact, adequate for describing the commonly observed LVT, so why can we not prove it mathematically?
In fact, the LVT is qualitatively described by approximate theories with mean-field-type interactions (where all particles interact with the same strength), dating back to the nineteenth century; see below and \cite{van73,Maxwell75}. 
However, for pair potentials without any symmetry, it has been proven rigorously only in the case of an attractive Kac type pair potential of the form $\gamma^d\varphi(\gamma r)$, where $d$ is the spatial dimension, in the infinite-range limit $\gamma\downarrow 0$, thus the need for a rigorous proof of the existence of the LVT in a continuum particle system with finite range (or rapidly decaying) pair interactions. 
We do this here for a simplified model for small $\gamma>0$.
We give a brief historical background of the LVT in Section \ref{sec:history}.

Mathematically speaking, we are interested in proving, for a continuous system of particles with stable pair interactions (see Appendix \ref{app:ruelle}) and no special symmetries, the existence of more than one infinite-volume Gibbs measure having different densities for some ranges of inverse temperature $\beta=1/T$ (setting Boltzmann's constant equal to $1$) and chemical potential $\lambda$ \cite{ruelle1971existence}.
For a system of particles in a region $\Lambda\subset\R^d$ with pair interactions $u(x-x')$, $x,x'\in\R^d$, the probability of having $N$ particles in a configuration $X_{\Lambda}=(x_1,\dots,x_N)\in\Lambda^N$ given a specified configuration in $\Lambda^c$, $Y_{\Lambda^c}$ (boundary condition), is given, in the grand-canonical Gibbs measure, by 
\begin{equation}
\label{eqn:grand-canonical finite volume Gibbs measure}
    \P^{\Lambda}_{\beta,\lambda}(\dd{X}_{\Lambda}\mid Y_{\Lambda^c})=\frac{1}{\Xi^{\Lambda}_{\beta,\lambda}}
    \frac{1}{N!}e^{-\beta[-\lambda N+U(X_\Lambda\mid Y_{\Lambda^c})]}\prod_{i=1}^N\dd{x}_i,
\end{equation}
where 
\begin{equation}
\label{eqn:particle model Hamiltonian}
    U(X_\Lambda\mid Y_{\Lambda^c})=\sum_{1\le i<j\le N}u(x_i-x_j)+\sum_{i,k}u(x_i-y_k)
\end{equation}
and $\Xi^{\Lambda}_{\beta,\lambda}$, the grand-canonical partition function, which depends on $Y_{\Lambda^c}$, is a normalizing factor.
We are interested in the behavior of macroscopic systems, idealized by taking the thermodynamic limit of $\Lambda\uparrow\R^d$.
The appropriate infinite-volume Gibbs measure is then characterized by the Dobrushin-Lanford-Ruelle (DLR) equation \cite{dobruschin1968description,lanford1969observables}, which specifies the conditional probability measure for any region $\Lambda\subset\R^d$ given a configuration in $\Lambda^c$ as in \eqref{eqn:grand-canonical finite volume Gibbs measure}.
The question then is whether this measure is unique.
If it is not unique, then we say that there is a coexistence of phases at that value of $\beta$ and $\lambda$.
When two translation-invariant, extremal Gibbs measures without any symmetry coexist but have different densities, $\rho_v<\rho_l$, we say that the system has an LVT, with $\rho_v$ and $\rho_l$ being the densities of the vapor and the liquid.

We can also look at the LVT from a macroscopic point of view: the average density $\rho$ for the finite system in the region $\Lambda$ is given by
\begin{equation}
    \rho^{\Lambda}(\beta,\lambda):=\frac{1}{\beta\abs{\Lambda}}\pdv{\lambda}\log\Xi^\Lambda_{\beta,\lambda},
\end{equation}
whose limit as $\Lambda\uparrow\R^d$, denoted by $\rho(\beta,\lambda)$, is a monotone increasing function of $\lambda$, which will have a discontinuity at the LVT, that is, at some value of $\lambda=\lambda_\ast$ where the two Gibbs measures (phases) coexist:
\begin{equation}
    \rho_v=\lim_{\lambda\uparrow\lambda_\ast}\rho(\beta,\lambda)<\rho_l=\lim_{\lambda\downarrow\lambda_\ast}\rho(\beta,\lambda).
\end{equation}

From a physical point of view, it is more natural to consider the LVT in the canonical ensemble where the density of the system, $\rho$, rather than its chemical potential, is specified. 
The LVT then corresponds to a linear segment for the canonical free energy \eqref{eq:Lebowitz-Penrose introduction} as a function of the density $\rho$, which, as explained in the next section, corresponds to the physical coexistence of the liquid and vapor phases.
It is in this ensemble in which Lebowitz--Penrose \cite{lebowitz1966rigorous} proved \eqref{eq:Lebowitz-Penrose introduction} for potentials of the form \eqref{eqn:introduction_pair-potential}.
The equivalence to the grand-canonical picture with two equilibrium states at some value of $\lambda$ is discussed in detail by Gates--Penrose \cite{gates1969vani,gates1970vanii,gates1970vaniii}. 
We shall skip over the fine details of this equivalence and use the physical picture in the rest of the introduction and then switch to the grand-canonical one in describing our results.

\begin{figure}
\centering
\includegraphics[width=0.5\linewidth]{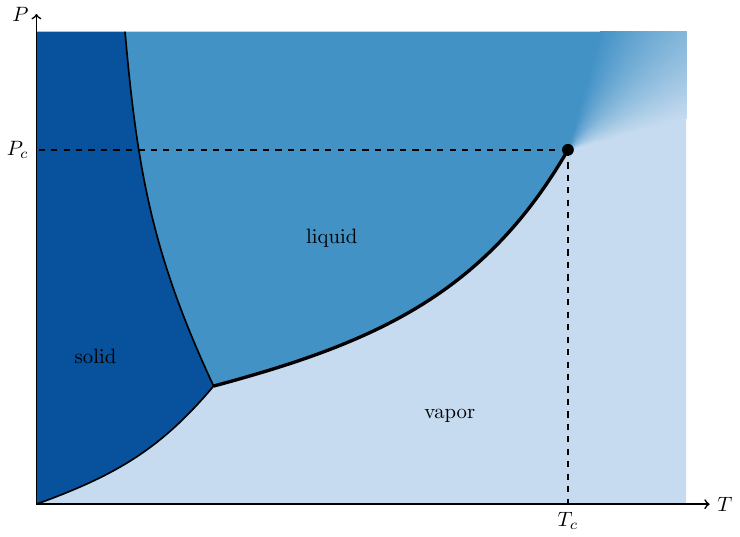}
\caption{A schematic phase diagram of a fluid in the temperature-pressure plane $(T,P)$.
There exists a critical point $(T_{c},P_{c})$ below which an LVT occurs, and the density (at which the free energy is minimized) jumps discontinuously when crossing the liquid-vapor transition line (thick black line).
The other lines correspond to fluid-solid transitions.}
\label{fig:phase-diagram}
\end{figure}

Before describing our results, we give a brief, selective history of the LVT problem.

\subsection{History of the liquid-vapor phase transition}
\label{sec:history}

\paragraph{Origin of the problem}

Experimental studies of the LVT go back to the beginning of the 19th century \cite{andrews1869xviii}, when physicists began looking for an expression for the pressure valid for densities beyond that of the very dilute gas, $p=\rho T$ (setting Boltzmann's constant equal to $1$), which would also describe the LVT in which liquid and vapor coexist at the same pressure.
In 1873, van der Waals in his doctoral thesis \cite{van73} derived heuristically an equation of state for the pressure as a function of temperature and density, which gave a qualitative understanding of the LVT,
\begin{equation}\label{eqn:introduction_van-der-Waals}
p(T,\rho)=\frac{T\rho}{1-\rho b}-\frac{1}{2}a\rho^{2}.
\end{equation}
Augmented by Maxwell's equal area construction in 1875 \cite{Maxwell75} (see Figure \ref{fig:maxwell-construction}), which ensures that the system is thermodynamically stable as described below, this equation gives, with suitable choices of empirical parameters $a,b>0$, a good qualitative description of the LVT observed in real systems \cite{Bernal33,Millot92}. 

The idea behind this equation of state (EOS) is that the force between atoms is strongly repulsive (hard-core like) at short distances and weakly attractive at large distances, which respectively give rise to the first and second terms on the RHS of \eqref{eqn:introduction_van-der-Waals}.
In fact, the second term can be derived by assuming a ``mean field'' attractive interaction between the particles, i.e., every pair of particles interact with a potential independent of the distance between them, whose strength is inversely proportional to the size of the system.
The first term in \eqref{eqn:introduction_van-der-Waals} can be written as $T/(\rho^{-1}-b)$, where the denominator represents the effective volume available to each particle.
It is, in fact, the exact pressure of a one-dimensional system of hard rods of diameter $b$ and can be considered an approximation for a strong short-range repulsion in higher-dimensional systems.

The canonical free energy density $f(T,\rho)$ is defined as the thermodynamic limit
\begin{equation}
    f(T,\rho):=-\lim_{\substack{N\to\infty,\Lambda\uparrow\R^d\\N/\abs{\Lambda}\to\rho}}\frac{T}{\abs{\Lambda}}\log Z(T,N,\Lambda),
\end{equation}
where 
\begin{equation}
    Z(T,N,\Lambda):=\frac{1}{N!}\int_{\Lambda^N}\dd{x}_1\dots\dd{x}_N e^{-\frac{1}{T} U(X_\Lambda)}
\end{equation}
is the canonical partition function, $\abs{\Lambda}$ is the volume of $\Lambda$, and $U(X_\Lambda)$ is the interaction potential of the $N$ particles in $\Lambda$, i.e., the first sum in \eqref{eqn:particle model Hamiltonian}.
Using the fact (see e.g. \cite{Ruelle69}) that the pressure is given by
\begin{equation}
\label{eqn:introduction_pressure}
    p(T,\rho)=\rho^2\pdv{\rho}\left[\frac{1}{\rho}f(T,\rho)\right],
\end{equation}
\eqref{eqn:introduction_van-der-Waals} is equivalent to the following expression for $f(T,\rho)$:
\begin{equation}
\label{eqn:introduction_free-energy}
    f(T,\rho)=-T\rho\log\frac{1-b\rho}{\rho}-\frac{1}{2}a\rho^2,
\end{equation}
up to a term independent of $\rho$.
The Maxwell construction thus corresponds to the Gibbs double tangent construction for $f(T,\rho)$; see Figure \ref{fig:constructions}.

The question then, as now, is how to derive the liquid-vapor phase transition from a ``realistic'' pair interaction between the particles. 

\begin{figure}
    \centering
    \begin{subfigure}[t]{0.46\textwidth}
        \centering
        \includegraphics[width=0.9\columnwidth]{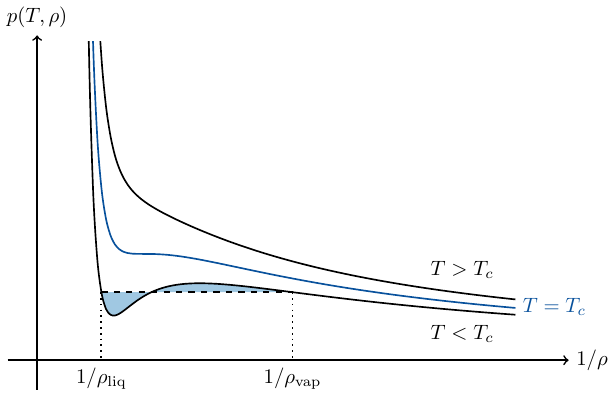}
        \caption{The solid line represents isotherms as given by \eqref{eqn:introduction_van-der-Waals} while the dotted line is the Maxwell construction giving equal areas to the solid color regions.
        This gives the coexistence of liquid and vapor phases at the same $T$ and $p$.}
        \label{fig:maxwell-construction}
    \end{subfigure}
    \hspace{12pt}
    \begin{subfigure}[t]{0.46\textwidth}
        \centering
        \includegraphics[width=0.9\columnwidth]{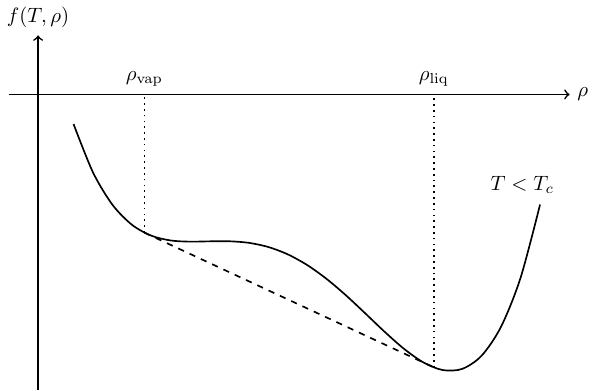}
        \caption{The Gibbs double tangent construction for $T<T_c$.}
        \label{fig:double-tangent}
    \end{subfigure}
    \caption{}
    \label{fig:constructions}
\end{figure}

\paragraph{Mathematical state-of-the-art}

A rigorous derivation of \eqref{eqn:introduction_van-der-Waals} was made by Kac, Uhlenbeck, and Hemmer (KUH) \cite{kac1963van}, who considered a one-dimensional system of hard rods with an attractive long-range pair interaction $\varphi_\gamma(\abs{x_i-x_j})$ of the Kac form,
\begin{equation}
\varphi_{\gamma}(r)=-\alpha\gamma e^{-\gamma r}.
\end{equation}
They proved that, in the thermodynamic limit, such a system of particles has an EOS $p_{\gamma}(T,\rho)$ which becomes, in the limit $\gamma\downarrow0$, the same as the van der Waals EOS \eqref{eqn:introduction_van-der-Waals} with Maxwell's construction \emph{built in}, rather than having to be added \emph{ad hoc} (as was the case before).
Lebowitz and Penrose (LP) \cite{lebowitz1966rigorous} generalized the result of KUH to arbitrary dimensions, proving, for systems with a pair potential of the form 
\begin{equation}
\label{eqn:introduction_pair-potential}
u_{\gamma}(r)=v(r)-\alpha\gamma^{d}\varphi(\gamma r),
\end{equation}
where $\varphi(r)\ge 0$, $\int_{\R^{d}}\dd{r}\varphi(r)=1$, and $v(r)$ has a hard-core part, the validity of the Gibbs double tangent construction, i.e., they proved, for both continuum and lattice systems, that
\begin{equation}\label{eq:Lebowitz-Penrose introduction}
f_+(T,\rho):=\lim_{\gamma\downarrow0}f_{\gamma}(T,\rho)=\CE\set{-\frac{1}{2}\alpha\rho^{2}+f_{0}(T,\rho)},
\end{equation}
where $f_{\gamma}(T,\rho)$ is the free energy density of the system with potential $u_{\gamma}(r)$, and $f_{0}(T,\rho)$ is the free energy density of the reference system with pair potential $v(r)$ containing a hard core.
The convex envelope (CE) of a function $\phi$ is the largest convex function smaller than $\phi$; see Figure~\ref{fig:double-tangent}.
The existence and good thermodynamic properties, e.g., convexity, of the free energies $f_{\gamma}(T,\rho)$ and $f_0(T,\rho)$ are assured by general results on the existence of the thermodynamic limit for systems with stable and tempered potentials; see Theorem \ref{theo:particle}. 

In one dimension, particle systems have no phase transition for potentials that decay faster than $r^{-2}$.
Hence, the free energy density $f_\gamma(T,\rho)$ is strictly convex in $\rho$ if $v(r)$ decays faster than $r^{-2}$ as $r\rightarrow\infty$ \cite{Ruelle69}, so the transition, as characterized by a linear portion of $f_\gamma$ as a function of $\rho$, only appears in the $\gamma\downarrow0$ limit.
This is not the case in $d\ge 2$ dimensions, where we expect not only a transition for $\gamma>0$ \cite{presutti2008scaling} but that the transition line should be, for small $\gamma>0$, close to its limiting value as $\gamma\downarrow0$.
This has been established by Presutti \cite{presutti2008scaling} for lattice gases on $\Z^{d}$ ($d\ge 2$) when $v(r)$ is just the single-site hard-core exclusion \cite{cassandro1996phase,bodineau1997phase,bovier1997low}.
However, for general lattice systems, without the particle-hole symmetry present in this case, there is no proof of the existence of an LVT at $\gamma>0$ close to that for $\gamma\downarrow 0$.
The situation is similar for continuum systems with stable pair potentials: the only proof of the existence of an LVT for continuum systems with finite-range potentials and no symmetries is given by Lebowitz, Mazel, and Presutti \cite{lebowitz1998rigorous,lebowitz1999liquid}, who had to resort to \emph{unphysical}, \emph{long-range four-body} repulsion to take the place of the short-range repulsive pair interactions to ensure stability of the system against collapse by the attractive Kac potential.
Their result, therefore, did not alleviate the need for proving the LVT for realistic pair potentials.
We also note here that Ruelle \cite{ruelle1971existence} proved the existence of a phase transition for the continuum two-species symmetric Widom-Rowlinson model and that Johansson \cite{johansson1995separation} proved a phase transition for a 1D continuum model whose pair potential decays like $r^{-\alpha}$, $\alpha\in(1,2)$.

\paragraph{The present work}

To obtain insight into the LVT for pair potentials of the form \eqref{eqn:introduction_pair-potential} at small $\gamma>0$ and its relation to the mean-field limit $\gamma\downarrow 0$, we study here a simplified model which we call the ``box model'' and which will be described in detail below.
In this model, we prove the existence of an LVT for small $\gamma>0$ given some properties of the free energy $f_{0}(T,\rho)$ of the reference system, the one in which the Kac potential is absent.
These properties are essentially the same as those used by LP to prove the LVT in the limit $\gamma\downarrow 0$, some of which are direct consequences of the existence of the thermodynamic limit for particle systems with superstable and tempered interactions \cite[\S3.1]{Ruelle69}; see Appendix \ref{app:ruelle}.
We show that, in the box model framework, an LVT occurs \emph{close} to the $\gamma\downarrow 0$ limit, in a precise sense, for sufficiently small $\gamma>0$.
Our result also applies to soft-core potentials, which LP did not consider.

The main benefit of the ``box model'' is that it is reflection positive. We remark that the idea of using reflection positivity to justify mean-field predictions has appeared before, in a different context, in the works~\cite{biskup2003rigorous,biskup2006mean} (see also~\cite[Section 4]{biskup2009reflection}).

\subsection{The boxed particle model and the box model}\label{sec:the box model}

In this paper, we introduce and study the \emph{box model}, which is a generalization of a particle model called the \emph{boxed particle model} obtained by modifying the pair potential in \eqref{eqn:introduction_pair-potential} that was studied by Lebowitz and Penrose in the mean-field limit $\gamma\downarrow 0$.
Let us first define the boxed particle model, which we will generalize to the box model in Section \ref{sec:box_model_def}.

The boxed particle model is defined by modifying \eqref{eqn:introduction_pair-potential} (primarily in the Kac potential part $\varphi$, but the short-range part will also be modified).
To do this, we partition $\R^d$ into a lattice of \emph{mesoscopic} cubes of side length $\gamma^{-1}$, and replace the long-range Kac interaction of LP with one that is constant for particles inside each of the cubes, and has a (possibly different) constant value for particles in adjacent cubes.
We will be quite general about the other interaction between particles given by $v(r)$ in \eqref{eqn:introduction_pair-potential}, and merely assume that it is superstable and tempered \cite{Ruelle69}.
However, we will neglect the interaction between different cubes due to $v(r)$.
A possible choice for $v(r)$ is a hard-core interaction, but we will be more general than that and allow interactions like the Lennard--Jones or Morse pair potentials (see Appendix \ref{app:ruelle} for more details).

More formally, we define the boxed particle model as a model for a system of particles interacting via the following pair interaction:
\begin{equation}
\label{eqn:boxed pair interaction}
  u_\gamma(x,y):=
  \begin{cases}
      v(x-y)-J_1\gamma^d & \mbox{if }\operatorname{Box}_\gamma(x)=\operatorname{Box}_\gamma(y)\\
      -J_2\gamma^d & \mbox{if }\operatorname{Box}_\gamma(x)\sim\operatorname{Box}_\gamma(y)\\
      0 & \mbox{otherwise}
  \end{cases}
\end{equation}
where $J_1>-2dJ_2$ and $J_2>0$ are constants, $\operatorname{Box}_\gamma(x):=\gamma^{-1}\floor{\gamma x}+[0,\gamma^{-1})^d$ denotes the unique cube in the mesoscopic lattice containing $x\in\R^d$, and the symbol $\sim$ means that two cubes are nearest neighbors, i.e., they share a $(d-1)$-dimensional face, and $v$ is a superstable and tempered potential (see Appendix \ref{app:ruelle} and \cite{Ruelle69}).
We fix the chemical potential in the boxed particle model to $\lambda+\frac12J_1\gamma^d$ (we add $\frac12J_1\gamma^d$ to simplify the notation in the ``spin'' model introduced in the next paragraph).

It is straightforward to check that the box model is equivalent to a ``spin'' model on the lattice $\mathbb Z^d$ with nearest neighbor interactions, where each point $v$ corresponds to a mesoscopic lattice cube as above, and the ``spin'' at $v$ is given by the density $\eta_v:=N_v/\gamma^{-d}$ of particles inside the cube corresponding to $v$.
Indeed, by integrating over the positions of the $N_v$ particles in each cube $v$, we find that the boxed particle model is equivalent to the following effective Hamiltonian on configurations of densities in the cubes:
\begin{equation}
\begin{split}
\label{eq:Hamiltonian_box}
H_{\lambda,\gamma}(\eta) 
:= &\gamma^{-d}\left[-\lambda\sum_v\eta_v- \frac{1}{2}J_{1}\sum_v \eta_v^{2} - J_{2}\sum_{v\sim w}\eta_v\eta_w + \sum_v f_{\gamma}(\eta_v)\right]\\
=&\gamma^{-d}\left[-\lambda\sum_v\eta_v- \frac{1}{2}(J_1+2d J_2)\sum_v \eta_v^{2} + \frac{1}{2}J_{2}\sum_{v\sim w}|\eta_v-\eta_w|^2 + \sum_v f_{\gamma}(\eta_v)\right],
\end{split}
\end{equation}
and
\begin{equation}
    f_\gamma(\eta_v):=-\frac{1}{\beta\gamma^{-d}}\log\frac1{(\eta_v\gamma^{-d})!}\int_{([0,\gamma^{-1})^{d})^{\eta_v\gamma^{-d}}}\dd{x}_1\cdots \dd{x}_{\eta_v\gamma^{-d}}\prod_{i<j}e^{-\beta v(x_i-x_j)}.
\end{equation}
is the canonical free energy density for particles in a cube interacting via the pair potential $v(\cdot)$ with free boundary conditions.

As noted earlier, we will study the boxed particle model in the grand-canonical ensemble, which is equivalent to the canonical ensemble \cite{Ruelle69}.
For the connection between the two ensembles in the limit $\gamma\downarrow0$, see \cite{gates1969vani,gates1970vanii,gates1970vaniii}.

\subsubsection{Formal description of the box model} \label{sec:box_model_def}

Let us now define the \emph{box model}, which is a slight generalization of the boxed particle model.
Specifically, we will relax the condition that $f_\gamma$ be the free energy of a particle model inside the mesoscopic cube, and merely require that it be a function that satisfies the conditions detailed in the rest of this section.
(Thus, whereas the boxed particle model is in fact a model of particles interacting via a pair potential, the box model is more general.)

When the pair potential $v(\cdot)$ has a hard core, $\eta$ takes values in a bounded subset of $\gamma^d\N$.
Another generalization we will make is to allow $\eta$ to take discrete or continuous values, as we will now describe.

Let $0<\gamma\le 1$ (the scale parameter). 
We allow both discrete and continuous non-negative values of $\eta:\Z^d\to S_\gamma$ where either $S_\gamma=[0,\infty)$, termed the \emph{continuous case}, or $S_\gamma=\gamma^d\Z\cap[0,\infty)$, termed the \emph{discrete case}.

Fix $\alpha, J_2>0$ (the coupling strengths) and $\beta>0$ (the inverse temperature)---these parameters will be held fixed throughout our arguments and will be omitted from the notation. Let $\lambda\in\R$ (the chemical potential). The formal Hamiltonian of the box model is given by
\begin{equation}\label{eq:Hamiltonian_box_rewritten}
H_{\lambda,\gamma}(\eta) 
\coloneqq \gamma^{-d}\left(\sum_{v}\left[-\lambda\eta_v- \frac{1}{2}\alpha\eta_v^{2} + f_{\gamma}(\eta_v)\right]+\frac{1}{2}J_{2}\sum_{v\sim w}\abs{\eta_{v}-\eta_{w}}^{2}\right),
\end{equation}
where $f_{\gamma}:S_\gamma\to\R\cup\{+\infty\}$ is a measurable function for each $\gamma$, satisfying the convergence assumption in Section \ref{sec:convergence_assumptions} (in particular, these conditions allow for $f_\gamma$ to be the free energy for a system of particles interacting via a superstable, tempered pair interaction; see Section \ref{sec:convergence_assumptions} for details, and so the boxed particle model introduced above is a special case of the box model). 
We note that $f_{\gamma}$ may depend on the fixed parameters $\alpha$, $J_2$, and $\beta$, but we will not make this dependence explicit.

Given an integer $L\ge 2$, the Gibbs measure of the box model on the discrete torus $\Lambda_L:=\mathbb Z^d/(L \mathbb Z)^d$ is
\begin{equation}\label{eq:box model finite-volume Gibbs measure periodic}
    \frac{1}{\Xi^{L, \per}_{\lambda,\gamma}}e^{-\beta H_{\lambda,\gamma}^{L,\per}(\eta)}\prod_{v\in\Lambda_L}\dd{\nu}_\gamma(\eta_v)
\end{equation}
where $\nu_\gamma$ is Lebesgue measure on $S_\gamma=[0,\infty)$ in the continuous case and $\nu_\gamma$ is the normalized counting measure $\gamma^d\sum_{\rho\in S_\gamma}\delta_\rho$ in the discrete case, where $H_{\lambda,\gamma}^{L,\per}$ is given by the  expression~\eqref{eq:Hamiltonian_box_rewritten} for $H_{\lambda,\gamma}$, changing the sums to run over $v\in\Lambda_L$ and $\{v,w\}\in E(\Lambda_L)$ (the edge set of the discrete torus graph), and where
\begin{equation}
\label{eq:box model finite-volume partition function periodic}
    \Xi^{L,\per}_{\lambda,\gamma}:=\int e^{-\beta H_{\lambda,\gamma}^\Lambda(\eta)}\prod_{v\in\Lambda}\dd{\nu}_\gamma(\eta_v)
\end{equation}
is the partition function (normalizing constant). One may similarly define Gibbs measures with free or prescribed boundary conditions.

\subsubsection{Convergence assumptions}\label{sec:convergence_assumptions}

We assume that the functions $(f_\gamma)_{0<\gamma\le 1}$ satisfy one of the following two conditions: (in the discrete case, for convenience in stating the assumption, we extend $f_{\gamma}$ to $[0,\infty)$ by a linear interpolation)
\begin{enumerate}
    \item Hard-core case: There is $\rho_{\cp}\in(0,\infty)$ and a continuous $f:[0,\rho_{\cp})\to\R$ such that:
\begin{enumerate}
    \item For every $\rho_0\in[0,\rho_{\cp})$,
    \begin{equation}\label{eq:uniform convergence below eta cp}
        \lim_{\gamma\downarrow0}f_{\gamma}(\rho) = f(\rho)\quad\text{uniformly in $\rho\in[0,\rho_0]$.}
    \end{equation}
    \item 
    \begin{equation}\label{eq:convergence at eta cp}
        \liminf_{\substack{\gamma\downarrow0\\\rho\to\rho_{\cp}}} f_{\gamma}(\rho) \ge \lim_{\rho\uparrow\rho_{\cp}} f(\rho)\in(-\infty,\infty]
    \end{equation}
    (the existence of the limit $\lim_{\rho\uparrow\rho_{\cp}} f(\rho)$ in $(-\infty,\infty]$ is part of the assumption).
    \item There is $\rho_{\max}\in(\rho_{\cp},\infty)$ such that $f_{\gamma}(\rho)=\infty$ for all $\rho>\rho_{\max}$ and $0<\gamma\le 1$. In addition, for every $\rho_1>\rho_{\cp}$,
    \begin{equation}\label{eq:uniform convergence above eta cp}
        \lim_{\gamma\downarrow0}\inf_{\rho\ge \rho_1}f_{\gamma}(\rho) = \infty.
    \end{equation}
\end{enumerate}
\item Soft-core case: There is a continuous $f:[0,\infty)\to\R$ such that:
\begin{enumerate}
    \item For every $\rho_0\in[0,\infty)$,
    \begin{equation}\label{eq:uniform convergence below infinity}
        \lim_{\gamma\downarrow0}f_{\gamma}(\rho) = f(\rho)\quad\text{uniformly in $\rho\in[0,\rho_0]$.}
    \end{equation}
    \item
    \begin{equation}\label{eq:growth at infinity}
        \alpha_{\max}:=\liminf_{\substack{\gamma\downarrow0\\\rho\to\infty}}\frac{f_{\gamma}(\rho)}{\frac{1}{2}\rho^2}\in(0,\infty].
    \end{equation}
\end{enumerate}
\end{enumerate}
To unify some of our later statements, we set $\alpha_{\max}\coloneqq\infty$ in the hard-core case.
We note that the above assumptions imply that the box model is well-defined for sufficiently small $\gamma$ and $\alpha<\alpha_\max$, in the sense made precise in Lemma~\ref{lem:model is well defined}.

\begin{remark}
\begin{enumerate}
\item In the hard-core case, it follows from~\eqref{eq:uniform convergence below eta cp} that
$\liminf_{\substack{\gamma\downarrow0\\\rho\to\rho_{\cp}}} f_{\gamma}(\rho) \le \liminf_{\rho\uparrow\rho_{\cp}} f(\rho)$. This complements the first inequality in~\eqref{eq:convergence at eta cp}, making it an equality.

\item
In the soft-core case, we note that~\eqref{eq:uniform convergence below infinity} and~\eqref{eq:growth at infinity} imply that
\begin{equation}\label{eq:growth of f at infinity}
    \liminf_{\rho\to\infty}\frac{f(\rho)}{\frac{1}{2}\rho^2}\ge\alpha_\max.
\end{equation}

\item
The above assumptions are satisfied for the boxed particle model, for which we recall that $f_{\gamma}(\rho)$ (see (\ref{eq:box model finite-volume partition function periodic})) is the free energy for a system of particles that take positions in the hypercubic volume $\gamma^{-d}$ and interact via a superstable and tempered pair potential.
This is a consequence of well-known results by Ruelle \cite{Ruelle69, ruelle1963classical}.
For more detailed references and a proof that these conditions are satisfied for superstable and tempered potentials, see Appendix \ref{app:ruelle} and Proposition \ref{prop:convergence_assumptions}.
In particular, the assumptions are verified when the interaction is a hard core, or the Lennard--Jones potential, or the Morse potential; see Proposition \ref{prop:superstable_examples}.
\end{enumerate}
\end{remark}

\subsubsection{Main results}
\label{sec:main results}

Our first result identifies the limiting grand-canonical pressure of the box model, showing that it coincides with the expression obtained by Gates--Penrose in the limit $\gamma\downarrow 0$ \cite{gates1969vani}. 
Define the {\it canonical mean-field free energy density} (following the nomenclature of \cite[(10.2.1.3)]{presutti2008scaling})
\begin{equation}\label{eq:mean-field free energy density}
    \phi_{\lambda}(\rho) 
:= -\lambda\rho-\frac{1}{2}\alpha\rho^{2}+ f(\rho).
\end{equation}

\begin{theorem}
\label{thm:comparison_with_GP}
In every dimension $d\ge 1$,
\begin{equation}  
\lim_{\substack{L\to\infty\\\gamma\downarrow0}}\frac{1}{\beta\gamma^{-d}|\Lambda_L|}\log\Xi^{L, \per}_{\lambda,\gamma} = -\inf_{\rho}\phi_{\lambda}(\rho),
\end{equation}
with the infimum over $\rho\in[0,\rho_\cp)$ in the hard-core case and over $\rho\in[0,\infty)$ in the soft-core case.
\end{theorem}
The same result (with the same proof) also holds for free boundary conditions, or for prescribed boundary conditions which are uniformly bounded as $L\to\infty, \gamma\downarrow 0$.

Lebowitz--Penrose~\cite{lebowitz1966rigorous} (canonical ensemble), followed by Gates--Penrose~\cite{gates1969vani} (grand-canonical ensemble), proved the existence of a liquid-vapor phase transition in the mean-field limit $\gamma\downarrow 0$ whenever the function $\phi_\lambda$ is non-convex. Our main result establishes the liquid-vapor phase transition at positive $\gamma$ (i.e., before taking the mean-field limit) in the box model.

\begin{theorem}\label{thm:main}
    Suppose the dimension $d\ge 2$. Suppose $\beta>0$ and $0<\alpha<\alpha_{\max}$ are such that $\phi_\lambda$ is non-convex (this property does not depend on $\lambda$). Then there exists $\gamma_0 > 0$ such that for all $0<\gamma\le\gamma_0$ there exists $\lambda(\gamma)$ for which the box model admits two distinct translation-invariant Gibbs measures that differ from each other in their value for the average density.
\end{theorem}

Moreover, as we formulate next, we show that the critical chemical potential and the densities of the liquid and vapor phases tend to their mean-field values as $\gamma\downarrow 0$.

In the hard-core case, the function $f$ is defined on the interval $[0,\rho_{\cp})$. We extend its domain to $[0,\rho_\cp]$ by setting
\begin{equation}\label{eq:f extension to rho cp}
    f(\rho_{\cp}):=\lim_{\rho\uparrow\rho_\cp} f(\rho),
\end{equation}
noting that $f(\rho_\cp)\in(-\infty,\infty]$ by~\eqref{eq:convergence at eta cp}. This also extends the domain of $\phi_\lambda$ to $[0,\rho_\cp]$, for all $\lambda$, via~\eqref{eq:mean-field free energy density}.

Suppose $\beta>0$ and $0<\alpha<\alpha_{\max}$ are such that $\phi_\lambda$ is non-convex. Let $\lambda_*\in\R$ be such that $\phi_{\lambda_*}$ attains its global minimum at (at least) two points $\rho_{*,-}<\rho_{*,+}$ with $\phi_{\lambda_\ast}$ non-constant on the interval $[\rho_{*,-},\rho_{*,+}]$; such a $\lambda_*$ necessarily exists by the non-convexity of $\phi_{\lambda}$ and \eqref{eq:convergence at eta cp}, in the hard-core case, or \eqref{eq:growth of f at infinity}, in the soft-core case.
Let
\begin{equation}
    \mathcal{M}:=\set{\rho\mid \phi_{\lambda_*}(\rho)=\inf_{\rho'}\phi_{\lambda_*}(\rho')}
\end{equation}
be the points where the global minimum is attained. Fix $\rho_{*,0}$ with $\rho_{*,-}<\rho_{*,0}<\rho_{*,+}$ and $\rho_{*,0}\notin\mathcal{M}$.

\begin{theorem}
\label{thm:detail}
    Suppose the dimension $d\ge 2$ and proceed in the above setup.
    Then there exists $\gamma_0 > 0$ such that, for all $0<\gamma\le\gamma_0$, there exists $\lambda_c(\gamma)$ for which the box model with $\lambda=\lambda_c(\gamma)$ admits two translation-invariant Gibbs measures $\P^\pm_{\lambda_c(\gamma), \gamma}$ and we have
    \begin{equation}\label{eq:convergence of critical chemical potential}
        \lim_{\gamma\downarrow0}\lambda_c(\gamma) = \lambda_*
    \end{equation}
    and, for every open set $U\subset\R$ containing $\mathcal{M}$,
    \begin{equation}\label{eq:density concentration}
        \lim_{\gamma\downarrow0}\P^-_{\lambda_c(\gamma), \gamma}(\eta_0\in U\cap[0,\rho_{*,0}))=\lim_{\gamma\downarrow0}\P^+_{\lambda_c(\gamma), \gamma}(\eta_0\in U\cap(\rho_{*,0},\infty))=1.
    \end{equation}
\end{theorem}

\begin{remark}\label{remark:remark after main theorem}
  In the (generic) case when $\mathcal{M}$ consists solely of the two points $\rho_{*,\pm}$, it follows from~\eqref{eq:density concentration} that the distribution of $\eta_0$ under $\P^\pm_{\lambda_c(\gamma), \gamma}$ converges in distribution to a delta measure at $\rho_{*,\pm}$.
\end{remark}

\subsection{The $(J,\omega)$-spin model and a condition for first-order phase transition}\label{J omega model section}

The main idea of the proof of Theorem \ref{thm:main} is to use the reflection positivity of the box model to derive a chessboard estimate, which allows us to use a result by Dobrushin and Shlosman \cite{shlosman1986method} guaranteeing the existence of the phase transition.
For the sake of completeness, we state and prove a version of the Dobrushin--Shlosman criterion that is suited to our result in Section \ref{sec:dobrushin_shlosman}.

Our analysis applies to a wider class of models. In this section, we introduce this class and state a condition that ensures a first-order phase transition therein.

\subsubsection{The $(J,\omega)$-spin model}
Given $J\ge0$ and a Borel measure $\omega$ on $\R^n$ with finite, positive total measure (i.e., $\omega(\R^n)\in(0,\infty)$), we give the name $(J,\omega)$-spin model to the $n$-component Gaussian free field with coupling constant $J$ and single-site measure $\omega$. This is the model on configurations $\eta:\Z^d\to\R^n$ whose formal Hamiltonian is
\begin{equation}\label{eq:(J,omega)-spin model Hamiltonian}
H(\eta):=J\sum_{v\sim w}\norm{\eta_{v}-\eta_{w}}^{2}.
\end{equation}

To apply the chessboard estimate, we consider the model on the discrete torus $\Lambda_L=\mathbb Z^d/(L \mathbb Z)^d$: Its finite-volume Hamiltonian is
\begin{equation}\label{Jw_ham}
H^L_J(\eta):=J\sum_{\{v,w\}\in E(\Lambda_L)}\norm{\eta_{v}-\eta_{w}}^{2}
\end{equation}
on configurations $\eta:\Lambda_L\to\R^n$, where $E(\Lambda_L)$ denotes the usual edge set of $\Lambda_L$, thought of as a graph, and $\|\cdot\|$ is the Euclidean norm (when convenient, we also regard such configurations as periodic functions on the entire $\Z^d$).
The corresponding finite-volume Gibbs measure is
\begin{equation}
  \P^{L,\per}_{J,\omega}(\dd{\eta})=\frac1{\Xi^{L,\per}_{J,\omega}} e^{-H^L_J(\eta)}\prod_{v\in\Lambda_L} \omega(\dd{\eta}_v)
  \label{prob_Jw}
\end{equation}
where $\Xi^{L,\per}_{J,\omega_\lambda}$ is the partition function:
\begin{equation}
  \Xi^{L,\per}_{J,\omega}:=\int e^{-H^L_J(\eta)}\prod_{v\in\Lambda_L}\omega(\dd{\eta}_v)
  .
  \label{partition_function}
\end{equation}

We will need (a bound on) the grand-canonical pressure
\begin{equation}
\label{eqn:spin-model_free-energy}
    \psi_{J,\omega}:=\liminf_{L\to\infty}\frac{1}{L^d}\log\Xi_{J,\omega}^{L,\per}.
\end{equation}
For later applications, we note the simple inequality
\begin{equation}
  \label{eqn:f-upper-bound}
  e^{\psi_{J,\omega}} \ge \sup_{S} e^{-dJ\diam(S)^{2}}\omega(S)
\end{equation}
where the supremum is over all measurable $S\subset\R^n$ and $\diam(S):=\sup_{x,y\in S}\|x-y\|$. Indeed, for each such $S$ we have
\begin{equation}
\label{eqn:partition-function-lower-bound}
  \Xi_{J,\omega}^{L,\per}
  \ge \int_{S^{\Lambda_L}} e^{-H^L_J(\eta)}\prod_v \omega(\dd{\eta}_v) \ge e^{-J\diam(S)^2|E(\Lambda_L)|}\omega(S)^{|\Lambda_L|}.
\end{equation}

\subsubsection{A continuous family of $(J,\omega)$-spin models}
\label{sec:continuous family of spin models}

We consider a \emph{continuous family of $(J,\omega)$-spin models}, indexed by $\lambda$ in an interval $[\lambda_-,\lambda_+]$.
Precisely, we consider a continuous function $\lambda\mapsto J_\lambda\ge 0$, and a function $\lambda\mapsto \omega_\lambda$ from $[\lambda_-,\lambda_+]$ to the set of Borel measures on $\R^n$ with finite, positive total measure that is continuous in distribution, in the sense that, for any converging sequence $\lambda_k\in[\lambda_-,\lambda_+]$ with $\lambda_k\to \lambda$ and any bounded continuous $g:\mathbb R^n\to \mathbb R$,
  \begin{equation}
  \label{eqn:good regions_continuity}
    \lim_{k\to\infty}\frac{1}{\omega_{\lambda_k}(\R^n)}\int g \dd{\omega}_{\lambda_k}=\frac{1}{\omega_{\lambda}(\R^n)}\int g \dd{\omega}_\lambda
    .
  \end{equation}

The box model defined in (\ref{eq:Hamiltonian_box}) corresponds to choosing $n=1$, letting
\begin{equation}
  J_{\lambda,\gamma}=\frac12\beta J_2\gamma^{-d}
\end{equation}
and letting 
\begin{equation}
  \omega_{\lambda,\gamma}(\dd{\rho})=e^{\beta\gamma^{-d}(-\lambda\rho-\frac12\alpha\rho^2+f_{\gamma}(\rho))}\nu_\gamma(\dd\rho),
\end{equation}
where $\nu_\gamma$ was defined after~\eqref{eq:box model finite-volume Gibbs measure periodic}.
These choices define a continuous family of $(J,\omega)$-spin models; see Section \ref{sec:deduction of main theorem}.

\subsubsection{First-order phase transition}

We now introduce further conditions that we will use for showing the existence of a first-order phase transition.
In these conditions, it is convenient to normalize the single-site measures using the pressures of~\eqref{eqn:spin-model_free-energy}: define, for each $\lambda\in[\lambda_-,\lambda_+]$,
\begin{equation}
  \tilde \omega_\lambda:=e^{-\psi_{J_\lambda,\omega_\lambda}}\omega_\lambda
  .
  \label{normalized_omega}
\end{equation}

\begin{assumption}\label{as:good regions}
    There exist closed sets $G_-,G_+\subset\R^n$ and $\theta_1,\theta_2,\theta_3\ge 0$ such that
  \begin{enumerate}
    \item The complement of $G_-\cup G_+$ has small measure: for all $\lambda\in[\lambda_-,\lambda_+]$,
    \begin{equation}
    \label{eqn:good regions_double-well}
	\tilde\omega_{\lambda}((G_-\cup G_+)^{c})\le\theta_1
	.
    \end{equation}
    \item $G_-$ and $G_+$ are separated: for all $\lambda\in[\lambda_-,\lambda_+]$,
    \begin{equation}
    \label{eqn:good regions_separation}
	e^{-\frac{1}{2}J_\lambda\dist(G_-, G_+)^2}(\tilde\omega_{\lambda}(G_-)\tilde\omega_{\lambda}(G_+))^{1/2} \le\theta_2
    \end{equation}
    with $\mathrm{dist}(G_-,G_+):=\inf_{x_-\in G_-,x_+\in G_+}\norm{x_--x_+}$.

    \item The region $G_\#$ is dominant at $\lambda_\#$:
    \begin{equation}
    \label{eqn:good regions_endpoints}
	\tilde\omega_{\lambda_\#}(G_\#^{c})\le\theta_3,
    \quad\#\in\set{-,+}.
    \end{equation}
  \end{enumerate}
\end{assumption}

\begin{theorem}
\label{thm:phase-transition}
  Suppose the dimension $d\ge 2$.
  For each $\epsilon\in(0,1/2)$ there exists $c(\eps, d)>0$ depending only on $\eps$ and the dimension $d$ such that the following holds.
  If Assumption~\ref{as:good regions} holds with $\max\{\theta_1,\theta_2,\theta_3\}\le c(\eps, d)$, then there exists $\lambda_{\mathrm c}\in(\lambda_-,\lambda_+)$ such that the $(J_{\lambda_c},\omega_{\lambda_c})$-spin model has two distinct translation-invariant Gibbs measures $\P_{-},\P_{+}$.
  Moreover,
  \begin{equation}\label{eq:P plus minus properties}
  \P_{\pm}(\eta_{0}\in G_{\pm})\ge 1-\epsilon
  .
  \end{equation}
\end{theorem}

\begin{remark}
    An explicit expression for the constant $c(\eps,d)$ can be obtained by examining the proofs in Section \ref{sec:Dobrushin-Shlosman_estimates}.
\end{remark}

\section{The Dobrushin--Shlosman criterion}
\label{sec:dobrushin_shlosman}

The main tool in the proof of Theorem~\ref{thm:phase-transition} is the Dobrushin--Shlosman criterion \cite{dobrushin1981phases},~\cite[Section 4]{shlosman1986method} for the existence of a first-order phase transition.
We present below a version of the criterion, adapted to our setting with the $(J,\omega)$-spin model.

\begin{theorem}[Dobrushin--Shlosman criterion]\label{thm:DS theorem}
Let $(J_\lambda,\omega_\lambda)_{\lambda\in[\lambda_-,\lambda_+]}$ be a continuous family of $(J,\omega)$-spin models.
For each $\lambda\in[\lambda_-,\lambda_+]$, let $\P_\lambda$ be a translation-invariant Gibbs measure of the $(J_\lambda,\omega_\lambda)$-spin model, such that the family of probability measures $(\P_\lambda)_{\lambda\in[\lambda_-,\lambda_+]}$ is tight.
Let $0<\eps<\frac{1}{2}$. 
Suppose that there exist disjoint closed subsets $G_-,G_+\subset \mathbb \R^n$ and $\delta_1,\delta_2>0$ such that
\begin{enumerate}
 \item $\delta_1+\delta_2 \leqslant 1-\frac \epsilon 2-\sqrt{1-\epsilon}$; \label{itm:DS theorem_constants}
 \item \label{itm:all lambda} for all $\lambda\in[\lambda_-,\lambda_+]$ and $v,w\in \mathbb Z^d$,
\begin{align}
&\mathbb P_\lambda(\eta_0\notin G_-\cup G_+)\le \delta_1,\label{eta_out}\\
&\mathbb P_\lambda(\eta_v\in G_-,  \eta_w\in G_+)\le \delta_2;\label{eta_+-}
\end{align}
\item $\P_{\lambda_\#}(\eta_0\in G_\#)\ge 1-\eps$ for $\#\in\{-,+\}$.
\label{itm:DS theorem_endpoints}
\end{enumerate}
Then, there exists $\lambda_c\in[\lambda_-,\lambda_+]$ such that the $(J_{\lambda_c},\omega_{\lambda_c})$-spin model admits two distinct translation-invariant Gibbs measures $\P_{\lambda_c,-}$ and $\P_{\lambda_c,+}$.
Moreover,
\begin{equation}
\P_{\lambda_c,\#}(\eta_0\in G_\#)\ge 1-\eps
,\quad
\#\in\set{-,+}.
\end{equation}
\end{theorem}

We give a proof of the theorem for completeness, starting with the following lemma.

\begin{lemma}
\label{claim:DS-ergodic} 
Define $\Delta_L:=\set{-L,-L+1,\dots,L-1,L}^d$ for each $L\ge 2$.
Let $\lambda\in[\lambda_-,\lambda_+]$. Let $\eta$ be sampled from $\P_\lambda$. 
Define
    \begin{equation}\label{eq:Pi def}
        \Pi_\#(\eta):=\lim_{L\rightarrow\infty}\frac{1}{\abs{\Delta_L}}\sum_{v\in\Delta_L}\indicator{G_\#}(\eta_v),\quad \#\in\{-,+\},
    \end{equation}
    which exist almost surely by the Birkhoff-Khinchin ergodic theorem.
    Under the assumptions of Theorem \ref{thm:DS theorem}, for all $\delta_3>0$,
    \begin{equation}\label{eq:large limiting density}
        \P_\lambda\left(\max\set{\Pi_-,\Pi_+}\ge 1-\delta_3\right)\ge 1-\frac{2(\delta_1+\delta_2)}{\delta_3}.
    \end{equation}
\end{lemma}

\begin{proof}
For each $L\ge 2$ and configuration $\eta$, define 
\begin{align}
    \Pi_{\#,L}(\eta){}&:=\frac{1}{\abs{\Delta_L}}\sum_{v\in\Delta_L}\indicator{G_\#}(\eta_v),\quad\#\in\set{-,+},\\
    \Psi_L(\eta){}&:=\frac{1}{\abs{\Delta_L}^2}\sum_{v,w\in\Delta_L}\indicator{(G_-^2\cup G_+^2)^c}(\eta_v,\eta_w).
\end{align}
Let $L\ge 2$.
By \eqref{eta_out} and \eqref{eta_+-},
\begin{equation}
\begin{multlined}
    \E_\lambda[\Psi_L]\le\frac{1}{\abs{\Delta_L}^2}\sum_{v,w\in\Delta_L}\left[
    \P_\lambda(\eta_v\not\in G_-\cup G_+)
    +\P_\lambda(\eta_w\not\in G_-\cup G_+)\right.
    \\
    \left.
    +\P_\lambda(\eta_v\in G_-,\eta_w\in G_+)
    +\P_\lambda(\eta_w\in G_-,\eta_v\in G_+)\right]
    \le 2 (\delta_1+\delta_2).
\end{multlined}
\end{equation}
Hence, by Markov's inequality,
\begin{equation}
\label{eqn:Dobrushin-Shlosman_Markov-inequality}
    \P_\lambda(\Psi_L<\delta_3)\ge 1-\frac{2(\delta_1+\delta_2)}{\delta_3}.
\end{equation}
Now, since $G_-$ and $G_+$ are disjoint, $\Pi_{-,L}+\Pi_{+,L}\le 1$, so
\begin{equation}
\label{eqn:Dobrushin-Shlosman_manipulation}
  \max\set{\Pi_{-,L},\Pi_{+,L}}
  \ge
  \max\set{\Pi_{-,L},\Pi_{+,L}}(\Pi_{-,L}+\Pi_{+,L})
  \ge
  \Pi_{-,L}^2+\Pi_{+,L}^2=1-\Psi_L.
\end{equation}
Combining \eqref{eqn:Dobrushin-Shlosman_Markov-inequality} and \eqref{eqn:Dobrushin-Shlosman_manipulation}, we get that
\begin{equation}
  \P_\lambda(\max\{\Pi_{-,L},\Pi_{+,L}\}\ge 1-\delta_3)\ge 1-\frac{2(\delta_1+\delta_2)}{\delta_3}.
\end{equation}
We conclude the proof using that $\Pi_{\#,L}\rightarrow\Pi_{\#}$ almost surely as $L\rightarrow\infty$.
\end{proof}

We are now ready to prove Theorem~\ref{thm:DS theorem}.

\begin{proof}[Proof of Theorem~\ref{thm:DS theorem}]
Define
\begin{equation}
    T_\#:=\left\{\lambda\in[\lambda_-,\lambda_+] \mid \P_\lambda(\eta_0\in G_\#)\ge 1-\eps\right\},\quad \#\in\{-,+\}.
\end{equation}
We consider two cases. 

First, suppose that $T_-\cup T_+\neq[\lambda_-,\lambda_+]$.
In other words, there exists $\lambda_c\in[\lambda_-,\lambda_+]$ such that
\begin{equation}\label{eq:P_lambda choice}
    \P_{\lambda_c}(\eta_0\in G_\#)< 1-\eps,\quad\#\in\{-,+\}.
\end{equation}
By considering the ergodic decomposition of $\P_{\lambda_c}$ (see, e.g., \cite[Theorem 14.17]{georgii2011gibbs}), we have that (recalling~\eqref{eq:Pi def})
\begin{equation}\label{eq:probaility as expectation}
    \P_{\lambda_c}(\eta_0\in G_\#) = \E_{\lambda_c}[\Pi_\#],\quad \#\in\{-,+\}.
\end{equation}
Using Item \ref{itm:DS theorem_constants} of the assumptions of the theorem, it is straightforward to check that there exists $\delta_3>0$ such that $(1-\delta_3)(1-2(\delta_1+\delta_2)/\delta_3)=1-\epsilon$ and $\delta_3\le\epsilon$.

\begin{claim}\label{claim:phase-transition_nonzero-prob}
\begin{equation}
\label{eqn:phase-transition_nonzero-prob}
    \P_{\lambda_c}(\Pi_\#>1-\delta_3)>0,\quad\#\in\set{-,+}.
\end{equation}
\end{claim}

\begin{proof}
We treat only the $\#=-$ case, as the other case is similar.
Suppose by contradiction that $\P_{\lambda_c}(\Pi_->1-\delta_3)=0$.
By Markov's inequality, Lemma~\ref{claim:DS-ergodic}, and the choice of $\delta_3$, we have that
\begin{equation}
\label{eqn:phase-transition_observable-lower-bound}
    \E_{\lambda_c}[\Pi_+]
    \ge\P_{\lambda_c}(\Pi_+\ge 1-\delta_3)(1-\delta_3)
    \ge\left(1-\frac{2(\delta_1+\delta_2)}{\delta_3}\right)(1-\delta_3)=1-\epsilon.
\end{equation}
On the other hand, from~\eqref{eq:P_lambda choice} and~\eqref{eq:probaility as expectation}, it follows that $\E_{\lambda_c}[\Pi_+]<1-\epsilon$, which contradicts \eqref{eqn:phase-transition_observable-lower-bound}.
\end{proof}

We now deduce that $\P_{\lambda_c}$ is not ergodic.
Indeed, suppose by contradiction that $\P_{\lambda_c}$ is ergodic.
As $\set{\Pi_\#>1-\delta_3}$ is a translation-invariant event, Claim \ref{claim:phase-transition_nonzero-prob} implies that $\mathbb P_{\lambda_c}(\Pi_\#>1-\delta_3)=1$, $\#\in\set{-,+}$.
Thus, there exists a configuration $\eta$ such that $\Pi_\#(\eta)>1-\delta_3$, $\#\in\set{-,+}$, so $1\ge \Pi_-(\eta)+\Pi_+(\eta)>2(1-\delta_3)\ge 2(1-\epsilon)>1$, a contradiction.
Therefore, $\mathbb P_{\lambda_c}$ is not ergodic, so its ergodic decomposition~\cite[Theorem 14.17]{georgii2011gibbs} contains two translation-invariant Gibbs measures $\P_{\lambda_c,-},\P_{\lambda_c,+}$ of the $(J_{\lambda_c},\omega_{\lambda_c})$-spin model 
such that
\begin{equation}
    \P_{\lambda_c,\#}(\eta_0\in G_\#)
    =\E_{\lambda_c,\#}[\Pi_\#]
    \ge 1-\delta_3\ge 1-\epsilon,\quad\#\in\{-,+\},
\end{equation}
establishing Theorem~\ref{thm:DS theorem} in this case.

Second, suppose that $T_-\cup T_+=[\lambda_-,\lambda_+]$. By Item \ref{itm:DS theorem_endpoints} of the assumptions, $\lambda_-\in T_-$ and $\lambda_+\in T_+$, so $T_-$ and $T_+$ are both non-empty. 
As $[\lambda_-,\lambda_+]$ is connected, it follows that there exists $\lambda_c\in\overline{T}_-\cap\overline{T}_+$ (where $\overline{A}$ denotes the closure of a set $A$). 
Let $(\lambda_{n,-})_{n\ge 1}\subset T_-$ satisfy $\lambda_{n,-}\to\lambda_c$. 
As the family $(\P_\lambda)_{\lambda\in[\lambda_-,\lambda_+]}$ is tight, it follows that $\P_{\lambda_{n_k,-}}\to \P_{\lambda_c,-}$ in distribution for some subsequence $(n_k)_{k\ge 1}$. The limit measure is a Gibbs measure of the $(J_{\lambda_c},\omega_{\lambda_c})$-spin model (see Proposition \ref{prop:convergence-of-Gibbs-measures}) and is clearly translation-invariant. 
Moreover, since $G_-$ is closed,
\begin{equation}
    \P_{\lambda_c,-}(\eta_0\in G_-)\ge \limsup_{k\rightarrow\infty} \P_{\lambda_{n_k,-}}(\eta_0\in G_-)\ge 1-\eps.
\end{equation}
By a symmetric argument, we obtain a translation-invariant Gibbs measure $\P_{\lambda_c,+}$ of the $(J_{\lambda_c},\omega_{\lambda_c})$-spin model satisfying $\P_{\lambda_c,+}(\eta_0\in G_+)\ge 1-\eps$. 
This completes the proof of the theorem.
\end{proof}

\section{Phase co-existence in the $(J,\omega)$-spin model}

In this section, we prove Theorem \ref{thm:phase-transition}.
Our proof is based on a criterion for the existence of first-order phase transitions due to Dobrushin and Shlosman \cite{shlosman1986method}; see Theorem \ref{thm:DS theorem}.
We re-prove this Theorem in Section \ref{sec:dobrushin_shlosman} in greater detail than the original reference.
In particular, our proof of this result applies to models with unbounded spin values.

\subsection{Reflection positivity and chessboard estimate}
\label{sec:chessboard_estimate}

We rely on the {\it chessboard estimate}, which follows from the reflection positivity of the $(J,\omega)$-spin model.
It is convenient to use \emph{reflections through hyperplanes intersecting edges}.

\subsubsection{Notation}
We continue to work on the discrete torus $\Lambda_L=\mathbb Z^d/(L \mathbb Z)^d$, restricting to even values of $L$.
For $1\le i\le d$ and $p\in \Z+\frac{1}{2}$, define the reflection of $\Lambda_L$ through the hyperplane orthogonal to direction $i$ at coordinate $p$,
\begin{equation}\label{eq:reflection transformation}
    (\tau_{i,p}(v))_j:=\begin{cases} 2p-v_j\ (\mathrm{mod}\ L)&\text{if } j=i\\
    v_j&\text{otherwise}\end{cases}.
\end{equation}

Now $\tau_{i,p}$ naturally acts on spin configurations $\eta$, and on functions of spin configurations. With a slight abuse of notation, we denote these actions by $\tau_{i,p}$ as well:
\begin{equation}
  (\tau_{i,p}\eta)_v:=\eta_{\tau_{i,p}(v)}
  ,\quad
  \tau_{i,p} f(\eta):=f(\tau_{i,p}(\eta))
  .
\end{equation}

Given $x\in\Z^{d}$ and $\vec{\ell}\in\Z_{\ge 0}^{d}$ such that $2(\vec\ell_i+1)$ divides $L$ for all $i$, define $R=R_{\vec{\ell},x}:=\prod_{i=1}^{d}[x_i,x_i+\vec{\ell}_i]\cap\Z^{d}$ as the box with corner $x$ and side lengths $\vec{\ell}$.
Let $T^R_{L}$ be the group of isomorphisms of $\Lambda_L$ that is generated by the reflections
\begin{equation}
  \bigcup_{i=1}^d \set{\tau_{i,p}\mid p\in x_i-\frac{1}{2}+(\vec{\ell}_i+1)\Z} \label{eq:P^R vert}
  .
\end{equation}
Note that
\begin{equation}\label{eq:size of T^R_L}
    \abs{T^R_L}=\prod_{i=1}^d\frac{L}{\vec{\ell}_i+1}.
\end{equation}
We also let $T^R$ be the group of isomorphisms of $\Z^d$ generated by the reflections in~\eqref{eq:P^R vert} (so that $T^R$ is infinite).

Recall that configurations of the $(J,\omega)$-spin model are functions $\eta:\Z^d\to\R^n$. An \emph{observable}, i.e., a measurable function $f:(\R^n)^{\Z^d}\to\R$, is called \emph{$R$-local} if $f(\eta)$ depends only on the restriction of $\eta$ to $R$.

\subsubsection{The chessboard estimate}
We are now ready to state the chessboard estimate.
Its proof is standard (see e.g. \cite[Theorem 5.8]{biskup2009reflection}, \cite[Theorem 10.11 and Remark 10.15]{friedli2017statistical}, or \cite[Section 2.7.1]{peled2019lectures}), and follows from the reflection positivity of the $(J,\omega)$-spin model, which is also standard (see e.g. \cite[Section 10.3.2]{friedli2017statistical}).
We will not reproduce either proof here.

\begin{lemma}[Chessboard estimate]
\label{lem:chessboard} 
Let $J>0$ and $\omega$ be a Borel measure on~$\R^n$ with finite, positive total measure.
Let $L\in\Z_{\ge 0}$ and $\vec{\ell}\in\Z_{\ge 0}^{d}$ satisfy that $2(\vec\ell_i+1)$ divides $L$ for all $i$. Let $x\in\Z^d$ and $R:=R_{\vec\ell,x}$. Let $A\subset T^{R}_L$ and $(f_{\tau})_{\tau\in A}$ be bounded $R$-local observables.
Then
\begin{equation}
\mathbb P_{J,\omega}^{L,\per}\left(\prod_{\tau\in A}\tau f_{\tau}\right)\le\prod_{\tau\in A}\|f_{\tau}\|_{J,\omega}^{R|L}
\end{equation}
where $\|\cdot\|_{J,\omega}^{R|L}$ is the chessboard seminorm of $f$, defined by
\begin{equation}
\label{eq:vertex chessboard norm def}
\norm{f}_{J,\omega}^{R\mid L} :=\Bigg[\mathbb P_{J,\omega}^{L,\per}\Bigg(\prod_{\tau\in T^{R}_{L}}(\tau f)\Bigg)\Bigg]^{1/\abs{T^{R}_{L}}},
\end{equation}
noting that $\abs{T^{R}_{L}}$ is given by~\eqref{eq:size of T^R_L}.
\end{lemma}

In our use of the chessboard estimate, it is convenient to pass to infinite volume. We first define the notion of a torus-limit Gibbs measure, taking care to take the limit on tori with highly divisible side lengths, to satisfy the assumptions of the chessboard estimate for arbitrary $\vec{\ell}\in\Z_{\ge 0}^d$.

\begin{definition}
\label{def:torus-limit}
We call a Gibbs measure $\P_{J,\omega}$ a \emph{torus-limit Gibbs measure of the $(J,\omega)$-spin model} if it is obtained as a limit in distribution along a subsequence of the (finite-volume) torus Gibbs measures $(\P^{k!,\per}_{J,\omega})_k$.
\end{definition}

\begin{corollary}[The chessboard estimate in the limit]
\label{cor:infinite-volume-chessboard-estimate}
Let $J\ge 0$ and $\omega$ be a Borel measure on~$\R^n$ with finite, positive total measure. Let $x\in\Z^d$, $\vec{\ell}\in\Z_{\ge 0}^{d}$ and $R:= R_{\vec\ell,x}$. Let $A\subset T^{R}$ be finite and $(f_{\tau})_{\tau\in A}$ bounded, $R$-local, lower semi-continuous observables.
Let $\P_{J,\omega}$ be a torus-limit Gibbs measure of the $(J,\omega)$-spin model.
Then,
\begin{equation}\label{eq:chessboard estimate infinite volume}
\P_{J,\omega}\left(\prod_{\tau\in A}\tau f_{\tau}\right)
\le \prod_{\tau\in A}\norm{f_{\tau}}_{J,\omega}^{R}
\end{equation}
where we define
\begin{equation}
    \norm{f_{\tau}}_{J,\omega}^{R}:=\limsup_{k\rightarrow\infty}\norm{f_{\tau}}_{J,\omega}^{R\mid k!}.
\end{equation}
\end{corollary}

\begin{proof}
By an immediate application of the Portmanteau theorem and the finite-volume chessboard estimate (Lemma \ref{lem:chessboard}).
\end{proof}

Though not needed for our results, we remark that it is desirable to extend~\eqref{eq:chessboard estimate infinite volume} to general periodic Gibbs measures, rather than just torus-limit Gibbs measures. Such an extension was demonstrated in~\cite{hadas2025columnar} for a different model. However, while the proof of~\cite{hadas2025columnar} applies in some generality, it does not cover the $(J,\omega)$-spin model for $\omega$ with non-compact support due to the fact that its interaction energy $\|\eta_v-\eta_w\|^2$ is unbounded.

\subsection{Proof of Theorem \ref{thm:phase-transition}}

We now prove Theorem \ref{thm:phase-transition} by verifying the assumptions of the Dobrushin--Shlosman criterion (Theorem \ref{thm:DS theorem}), where we take $\P_{\lambda}$ to be a torus-limit Gibbs measure (see Remark \ref{rem:existence-torus-limit}) of the $(J_\lambda,\omega_\lambda)$-spin model, $\lambda\in[\lambda_-,\lambda_+]$.
This requires verifying the tightness of this family of measures as well as three probabilistic estimates.
In Section \ref{sec:torus-limit-measures-tight}, we resolve the first issue by taking advantage of the continuity of the mappings $\lambda\mapsto J_\lambda$ and $\lambda\mapsto\omega_\lambda$.
In Section \ref{sec:Dobrushin-Shlosman_estimates}, we use Assumption \ref{as:good regions} to tackle the second issue.

\subsubsection{Tightness of torus-limit Gibbs measures}
\label{sec:torus-limit-measures-tight}
It is convenient to use the following notation in this section: Given a Borel measure $\omega$ on $\R^n$ with finite, positive total measure, we denote its normalized probability measure by
\begin{equation}
\label{eqn:normalized probability measure}
    \bar{\omega}:=\frac{\omega}{\omega(\R^n)}.
\end{equation}

We formulate the tightness property of the torus-limit Gibbs measures in slightly greater generality than needed.

\begin{proposition}[Tightness of torus-limit Gibbs measures]
\label{prop:torus-limit-measures-tight}
Let $I$ be an arbitrary index set.
Suppose that $\set{J_{i}\mid i\in I}$ is a bounded subset of $\R_{\ge 0}$ and $\set{\omega_{i}\mid i\in I}$ a family of Borel measures on $\R^n$ with finite, positive total measure, such that the family $\set{\bar{\omega}_{i}\mid i\in I}$ is tight.
For each $i\in I$, let $\P_i$ be a torus-limit Gibbs measure of the $(J_i,\omega_i)$-spin model.
Then, the family of measures $\set{\P_i\mid i\in I}$ is tight.
\end{proposition}

Proposition \ref{prop:torus-limit-measures-tight} is sufficient for our purpose.
Indeed, recall that we are considering a continuous family $(J_\lambda,\omega_\lambda)_{\lambda\in[\lambda_-,\lambda_+]}$ of $(J,\omega)$-spin models.
The continuity of the mappings $\lambda\mapsto J_\lambda$ and $\lambda\mapsto\omega_\lambda$ imply that the set $\set{J_\lambda\mid\lambda\in[\lambda_-,\lambda_+]}$ is bounded and that the family of probability measures $\set{\bar{\omega}_\lambda\mid\lambda\in[\lambda_-,\lambda_+]}$, as the continuous image of a compact set, is compact, thus tight.

\smallskip
We begin the proof of Proposition~\ref{prop:torus-limit-measures-tight} by establishing the tightness of the family of finite-volume torus Gibbs measures of the $(J_i,\omega_i)$-spin models, $i\in I$.

\begin{lemma}
\label{lem:torus-measures-tight-family}
Let $I$ be an arbitrary index set.
Suppose that $\set{J_{i}\mid i\in I}$ is a bounded subset of $\R_{\ge 0}$ and $\set{\omega_{i}\mid i\in I}$ a family of Borel measures on $\R^n$ with finite, positive total measure, such that the family $\set{\bar{\omega}_{i}\mid i\in I}$ is tight.
Then, the set of (finite-volume) torus Gibbs measures 
\begin{equation}
\set{\P_{J_i,\omega_i}^{L,\per}\mid L\ge 2 \text{ even, }i\in I}
\end{equation}
is tight.
\end{lemma}

\begin{proof}
Let $\epsilon>0$.
We will construct a family $(K_{v,\epsilon})_{v\in\Z^d}$ of compact subsets of $\R^n$ such that 
\begin{equation}
\P^{L,\per}_{J_{i},\omega_{i}}(\eta_v \in K_{v,\epsilon} \text{ for all }v\in\Z^d) \ge 1 -\epsilon,\quad\text{for all } L\ge2\text{ even and }i\in I.
\end{equation}
This completes the proof of the lemma, since the product $\prod_{v\in\Z^{d}}K_{v,\epsilon}$ of these compact sets is a compact subset of $(\R^n)^{\Z^d}$ by Tychonoff's theorem.

We proceed with the construction.
By a union bound, it suffices to choose the sets $(K_{v,\epsilon})_{v\in\Z^d}$ such that 
\begin{equation}
\label{eqn:tightness-union-bound}
\begin{multlined}
\sum_{v\in\Z^{d}} \P^{L,\per}_{J_{i},\omega_{i}}(\eta_v \notin K_{v,\epsilon})
\le\epsilon,
\quad\text{for all } L\ge2\text{ even and }i\in I.
\end{multlined}
\end{equation}
Fix $i\in I$ and an even $L\ge 2$.
By the chessboard estimate (Lemma \ref{lem:chessboard}), for any $v\in\Z^{d}$ and compact subset $K\subset\R^{n}$, we have that
\begin{equation}
\P^{L,\per}_{J_{i},\omega_{i}}(\eta_v \notin K) \le \P^{L,\per}_{J_{i},\omega_{i}}(\eta_w \notin K\text{ for all }w\in\Lambda_{L})^{1/L^d},
\end{equation}
where
\begin{equation}
\P^{L,\per}_{J_{i},\omega_{i}}(\eta_w \notin K\text{ for all }w\in\Lambda_{L}) 
=\frac{1}{\Xi_{J_{i},\omega_{i}}^{L,\per}}
\int_{(K^{c})^{\Lambda_{L}}} e^{-H^{L}_{J_{i}}(\eta)} \prod_{w\in\Lambda_{L}} \omega_{i}(\dd{\eta}_w).
\end{equation}
Bounding $H^{L}_{J_{i}}(\eta)\ge0$ and using the lower bound \eqref{eqn:partition-function-lower-bound} on $\Xi_{J_{i},\omega_{i}}^{L,\per}$, we conclude that
\begin{equation}
\label{eqn:torus-measures-tight-family_probability-bound}
\P^{L,\per}_{J_{i},\omega_{i}}(\eta_v \notin K)
\le \frac{\omega_{i}(K^{c})}{e^{-dJ_{i}\diam(S)^{2}}\omega_{i}(S)}
=\frac{\bar{\omega}_{i}(K^{c})}{e^{-dJ_{i}\diam(S)^{2}}\bar{\omega}_{i}(S)}
\end{equation}
for all measurable $S\subset\R^{n}$.
By the tightness of the family of probability measures $\set{\bar{\omega}_i\mid i\in I}$, we may choose $S$ to be a compact subset of $\R^n$ such that $\inf_{i\in I}\bar{\omega}_i(S)\ge 1/2$.
Next, fix a function $\delta:\Z^{d}\rightarrow\R_{>0}$ such that $\sum_{v\in\Z^{d}}\delta(v)\le 1$.
Using tightness again and recalling that the set $\set{J_i\mid i\in I}$ is bounded, we choose, for each $v\in\Z^d$, a compact set $K_{v,\epsilon}\subset\R^{n}$ such that
\begin{equation}
\label{eqn:torus-measures-tight-family_choice}
\bar{\omega}_{i}(K_{v,\epsilon}^c) 
\le \epsilon\delta(v)
e^{-d\diam(S)^{2}\sup_{i'\in I}J_{i'}}
\inf_{i'\in I}\bar{\omega}_{i'}(S),\quad\text{for all }i\in I.
\end{equation}
It remains to check that the family $(K_{v,\epsilon})_{v\in\Z^d}$ thus chosen satisfies \eqref{eqn:tightness-union-bound}, but this is an immediate consequence of \eqref{eqn:torus-measures-tight-family_probability-bound}, \eqref{eqn:torus-measures-tight-family_choice}, and the fact that $\sum_{v\in\Z^d}\delta(v)\le 1$.
\end{proof}

\begin{remark}
\label{rem:existence-torus-limit}
    It follows immediately from Lemma \ref{lem:torus-measures-tight-family} that every $(J,\omega)$-spin model admits at least one torus-limit Gibbs measure.
    In turn, this justifies our choice of $\P_\lambda$ as a torus-limit Gibbs measure of the $(J_\lambda,\omega_\lambda)$-spin model for each $\lambda\in[\lambda_-,\lambda_+]$, as made at the beginning of this subsection.
\end{remark}

We now deduce Proposition \ref{prop:torus-limit-measures-tight}.

\begin{proof}[Proof of Proposition \ref{prop:torus-limit-measures-tight}]
By Lemma \ref{lem:torus-measures-tight-family}, the set of torus Gibbs measures 
\begin{equation}
\mathcal{G}=\set{\P^{L,\per}_{J_{i},\omega_{i}}\mid L\ge 2\text{ even},i\in I}
\end{equation}
is tight, namely (using Prokhorov's theorem \cite[Theorem 5.1]{billingsley2013convergence}), its closure $\overline{\mathcal{G}}$ is compact.
Further, for each $i\in I$, $\P_{i}$ is the limit in distribution along a subsequence of the torus Gibbs measures $(\P^{k!,\per}_{J_{i},\omega_{i}})_{k\ge 1}$ by definition.
Thus, the set of torus-limit Gibbs measures $\set{\P_{i}\mid i\in I}$ is a subset of $\overline{\mathcal{G}}$. 
As $\overline{\mathcal{G}}$ is compact, $\set{\P_{i}\mid i\in I}$ has compact closure, hence a tight family of measures (again using Prokhorov's theorem \cite[Theorem 5.2]{billingsley2013convergence}).
\end{proof}

\subsubsection{Verification of key estimates}
\label{sec:Dobrushin-Shlosman_estimates}

We now verify the three estimates in the Dobrushin--Shlosman criterion (Theorem \ref{thm:DS theorem}) with the help of Assumption \ref{as:good regions}.

We start with \eqref{eta_out} of Item \ref{itm:all lambda}.

\begin{proposition}
\label{prop:Dobrushin-Shlosman_double-well}
Under Assumption \ref{as:good regions}, for all $\lambda\in[\lambda_{-},\lambda_{+}]$, if $\P_\lambda$ is a torus-limit Gibbs measure of the $(J_\lambda,\omega_\lambda)$-spin model, then
\begin{equation}
\P_{\lambda}(\eta_{0}\not\in G_{-}\cup G_{+})\le\theta_1.
\end{equation}
\end{proposition}

\begin{proof}
Let $R:=\set{0}^{d}$ denote the origin of $\Z^{d}$. 
Consider the $R$-local event $E:=\set{\eta\in\Omega\mid \eta_{0}\in (G_{-}\cup G_{+})^{c}}$.
By Corollary \ref{cor:infinite-volume-chessboard-estimate}, 
\begin{equation}
\label{eqn:Dobrushin-Shlosman_double-well_chessboard-estimate}
\P_{\lambda}(E)
\le\limsup_{k\rightarrow\infty}\norm{E}^{R\mid k!}_{J_\lambda,\omega_\lambda},
\end{equation}
where the chessboard seminorm
\begin{equation}
\label{eqn:Dobrushin-Shlosman_double-well_chessboard-seminorm}
\norm{E}^{R\mid k!}_{J_\lambda,\omega_\lambda}
=\P^{k!,\per}_{J_{\lambda},\omega_{\lambda}}\Bigg(\bigcap_{\tau\in T^{R}_{k!}}\tau(E)\Bigg)^{1/(k!)^{d}}
\end{equation}
is well-defined for all $k\ge 2$.
To estimate the probability on the RHS of \eqref{eqn:Dobrushin-Shlosman_double-well_chessboard-seminorm}, we use the following simple bound on the partition function of the $(J_{\lambda},\omega_{\lambda})$-spin model in $\Lambda_{k!}$.
Let $\nu>0$.
By \eqref{eqn:spin-model_free-energy}, there exists $k_\nu\in\N$ such that, for all $k\ge k_\nu$,
\begin{equation}
\label{eqn:Dobrushin-Shlosman_double-well_partition-function-lower-bound}
    \frac{1}{(k!)^{d}}\log\Xi_{J_\lambda,\omega_\lambda}^{k!,\per}
    \ge\psi_{J_\lambda,\omega_\lambda}-\nu.
\end{equation}
On the other hand, the event $\bigcap_{\tau\in T^{R}_{k!}}\tau(E)$ consists of configurations with all spins in $(G_{-}\cup G_{+})^{c}$.
Thus,
\begin{equation}
\label{eqn:Dobrushin-Shlosman_double-well_final}
\P^{k!,\per}_{J_{\lambda},\omega_{\lambda}}\Bigg(\bigcap_{\tau\in T^{R}_{k!}}\tau(E)\Bigg)
\le\frac{\omega_{\lambda}((G_{-}\cup G_{+})^{c})^{
(k!)^{d}}}{e^{(\psi_{J_\lambda,\omega_\lambda}-\nu)(k!)^{d}}}
\le (e^\nu \theta_1)^{(k!)^d},
\end{equation}
where we used \eqref{eqn:good regions_double-well} of Assumption \ref{as:good regions} in the last inequality.
Inserting \eqref{eqn:Dobrushin-Shlosman_double-well_final} into \eqref{eqn:Dobrushin-Shlosman_double-well_chessboard-seminorm}, recalling \eqref{eqn:Dobrushin-Shlosman_double-well_chessboard-estimate}, and taking $\nu\rightarrow0$ complete the proof.
\end{proof}

Next, we verify Item \ref{itm:DS theorem_endpoints}.

\begin{proposition}
\label{prop:endpoint}
Under Assumption \ref{as:good regions}, for $\#\in\set{-,+}$, if $\P_{\lambda_\#}$ is a torus-limit Gibbs measure of the $(J_{\lambda_\#},\omega_{\lambda_\#})$-spin model, then
\begin{equation}
\label{eqn:endpoint}
\P_{\lambda_{\#}}(\eta_{0}\in G_{\#})\ge1-\theta_3.
\end{equation}
\end{proposition}

\begin{proof}
The proof is analogous to that of Proposition \ref{prop:Dobrushin-Shlosman_double-well}: replacing $(G_{-}\cup G_{+})^{c}$ by $G_{\#}^{c}$ and using \eqref{eqn:good regions_endpoints} of Assumption \ref{as:good regions}, we get that
\begin{equation}
\P_{\lambda_{\#}}(\eta_{0}\in G_{\#}^{c})\le\theta_3,
\end{equation}
which immediately implies \eqref{eqn:endpoint}.
\end{proof}

Finally, we verify \eqref{eta_+-} of Item \ref{itm:all lambda}.

\begin{proposition}
\label{prop:Dobrushin-Shlosman_interface}
For all $\delta_2>0$, there exist constants $\theta_1(\delta_2),\theta_2(\delta_2)>0$ such that, if Assumption \ref{as:good regions} holds with $0<\theta_i\le\theta_i(\delta_2)$, $i=1,2$, then, for all $\lambda\in[\lambda_{-},\lambda_{+}]$, torus-limit Gibbs measure $\P_\lambda$ of the $(J_\lambda,\omega_\lambda)$-spin model, and $v,w\in\Z^{d}$, 
\begin{equation}
\P_{\lambda}(\eta_{v}\in G_{-},\eta_{w}\in G_{+})\le\delta_2.
\end{equation}
\end{proposition}

\begin{remark}
    It is possible to extract quantitative estimates for the constants $\theta_i(\delta_2)$, $i=1,2$, as promised by the proposition from a fully explicit chessboard-Peierls argument, in terms of the smallness of certain series.
    We do not attempt this here.
\end{remark}

\begin{proof}
The proof is by a standard chessboard-Peierls argument, which we do not belabor here and refer the reader to the literature \cite{frohlich1978phase} for detailed implementations.
Intuitively, the conditions $\eta_v\in G_-$, $\eta_w\in G_+$ imply the existence of geometric interfaces in $\Z^d$, consisting of edges connecting vertices which are costly neighbors due to their spin values.
Of particular relevance to the chessboard-Peierls argument are the set of boundary edges of the connected component of $\set{u\in\Z^d\mid\eta_u\in G_-}$ containing $v$, where each edge connects a spin in $G_-$ and another in $G_-^c=G_+\cup(G_-\cup G_+)^c$, and an analogous set of edges for $w$.
The probability of observing these costly edge events simultaneously is then controlled using the chessboard estimate.

In our case, we consider the following edge events.
Let $R:=\set{v_1,v_2}$ be an edge of $\Z^d$.
Define
\begin{align}
\label{eqn:Dobrushin-Shlosman_interface_edge-event}
E_{v_1,v_2}^{-,+}&:=\set{\eta\in\Omega\mid\eta_{v_1}\in G_{-},\eta_{v_2}\in G_{+}},\\
E_{v_1,v_2}^{-,0}&:=\set{\eta\in\Omega\mid\eta_{v_1}\in G_{-},\eta_{v_2}\in (G_{-}\cup G_{+})^{c}},
\end{align}
and similarly $E_{v_1,v_2}^{+,-}$ and $E_{v_1,v_2}^{+,0}$.
In view of the above outline of the chessboard-Peierls argument, to prove the proposition, it suffices to show that the chessboard seminorm of each of the four edge events can be made arbitrarily small by taking $\delta$ correspondingly small.
We establish this for $E_{v_1,v_2}^{-,+}$ and $E_{v_1,v_2}^{-,0}$ in individual claims below, and note that symmetric arguments yield identical bounds for $E_{v_1,v_2}^{+,-}$ and $E_{v_1,v_2}^{+,0}$.

We start with $E_{v_1,v_2}^{-,0}$ which is simpler to deal with.

\begin{claim}
Under Assumption \ref{as:good regions}, for all $\lambda\in[\lambda_-,\lambda_+]$,
\begin{equation}
\label{eqn:Dobrushin-Shlosman_interface_edge-event_-,0 probability}
\norm{E^{-,0}_{v_1,v_2}}^R_{J_\lambda,\omega_\lambda}\le\theta_1.
\end{equation}
\end{claim}

\begin{proof}
Observe that the edge event $E^{-,0}_{v_1,v_2}$ is contained in the single-vertex event that $\eta_{v_2}\in(G_-\cup G_+)^c$.
Recall from the proof of Proposition \ref{prop:Dobrushin-Shlosman_double-well} that the chessboard seminorm of the latter is bounded by $\delta$ (the proof there is written for when $v_2$ is the origin, but this is inconsequential).
The monotonicity of the chessboard seminorm then implies \eqref{eqn:Dobrushin-Shlosman_interface_edge-event_-,0 probability}.
\end{proof}

The treatment of $E_{v_1,v_2}^{-,+}$ is complicated by the fact that it is not an open event, i.e., an open subset of $\Omega$, so its indicator function does not fulfill the lower semi-continuity requirement of the chessboard estimate in the limit (Corollary \ref{cor:infinite-volume-chessboard-estimate}).
To overcome this technical nuisance, we construct below a sequence of open events containing $E_{v_1,v_2}^{-,+}$ satisfying suitable properties.

\begin{claim}
\label{clm:Dobrushin-Shlosman_interface_edge-event_-,+}
There exists a decreasing sequence of open events $(E_{v_1,v_2;j}^{-,+})_{j\ge 1}$ such that
\begin{enumerate}
    \item $\bigcap_{j=1}^\infty E_{v_1,v_2;j}^{-,+}=E_{v_1,v_2}^{-,+}$; \label{itm:Dobrushin-Shlosman_interface_edge-event_-,+_intersection}
    \item under Assumption \ref{as:good regions}, for all $\lambda\in[\lambda_-,\lambda_+]$, 
    \begin{equation}
        \limsup_{j\rightarrow\infty}\norm{E_{v_1,v_2;j}^{-,+}}^R_{J_\lambda,\omega_\lambda}\le\theta_2^2.
    \end{equation}
    \label{itm:Dobrushin-Shlosman_interface_edge-event_-,+_seminorm}
\end{enumerate}
\end{claim}

\begin{proof}
For each $j\ge 1$, define the open $2^{-j}$-extensions of a set $A\subset\R^n$ by
\begin{equation}
    O_j(A):=A+(-2^{-j},2^{-j})^n,
\end{equation}
where $+$ denotes the Minkowski sum.
Observe that if $A$ is closed, then
\begin{equation}\label{eq:intersection of extensions}
    A=\bigcap_{j=1}^\infty O_{j}(A).
\end{equation}
Using the above notation, define, for each $j\ge 1$,
\begin{equation}
\label{eqn:-+bound1}
    E_{v_1,v_2;j}^{-,+}:=\set{\eta\in\Omega\mid\eta_{v_1}\in O_j(G_{-}),\eta_{v_2}\in O_{j}(G_{+})},
\end{equation}
which forms a decreasing sequence of open events.
Property \ref{itm:Dobrushin-Shlosman_interface_edge-event_-,+_intersection} immediately follows from \eqref{eq:intersection of extensions}.
To prove Property \ref{itm:Dobrushin-Shlosman_interface_edge-event_-,+_seminorm}, fix $j\ge 1$ and consider the chessboard seminorm
\begin{equation}\label{eq:chessboard norm bound on kappa events}
    \norm{E_{v_1,v_2;j}^{-,+}}^R_{J_\lambda,\omega_\lambda}
    =\limsup_{k\rightarrow\infty}\norm{E_{v_1,v_2;j}^{-,+}}^{R\mid k!}_{J_\lambda,\omega_\lambda}
    =\limsup_{k\rightarrow\infty}\P^{k!,\per}_{J_{\lambda},\omega_{\lambda}}\Bigg(\bigcap_{\tau\in T^{R}_{k!}}\tau(E^{-,+}_{v_1,v_2;j})\Bigg)^{2/(k!)^{d}}
\end{equation}
Note that the probability on the RHS of \eqref{eq:chessboard norm bound on kappa events} is well-defined for all $k\ge 4$.
To bound this probability, let $\nu>0$ be arbitrary and recall the lower bound \eqref{eqn:Dobrushin-Shlosman_double-well_partition-function-lower-bound} on the partition function of the $(J_\lambda,\omega_\lambda)$-spin model in $\Lambda_{k!}$, which holds for all $k\ge k_\nu$ for some $k_\nu\ge 4$.
Taking into account that the constraint $\bigcap_{\tau\in T^{R}_{k!}}\tau(E^{-,+}_{v_1,v_2;j})$ forces exactly $1/2d$ of the edges in $\Lambda_{k!}$ to connect a spin in $O_j(G_{-})$ to another in $O_j(G_{+})$, we get that
\begin{equation}
    \P^{k!,\per}_{J_{\lambda},\omega_{\lambda}}\Bigg(\bigcap_{\tau\in T^{R}_{k!}}\tau(E^{-,+}_{v_1,v_2;j})\Bigg)^{2/(k!)^{d}}
    \le
    \frac{\omega_{\lambda}(O_j(G_{-}))\omega_{\lambda}(O_j(G_{+}))e^{-J_{\lambda}\dist(O_{j}(G_{-}),O_{j}(G_{+}))^{2}}}
    {e^{2(\psi_{J_\lambda,\omega_\lambda}-\nu)}}
\end{equation}
for all $k\ge k_\nu$.
Taking the limit superior $k\to\infty$ and then $\nu\to0$, we conclude that 
\begin{equation}
    \norm{E_{v_1,v_2;j}^{-,+}}^R_{J_\lambda,\omega_\lambda}
    \le
    \frac{\omega_{\lambda}(O_j(G_{-}))\omega_{\lambda}(O_j(G_{+}))e^{-J_{\lambda}\dist(O_{j}(G_{-}),O_{j}(G_{+}))^{2}}}
    {e^{2\psi_{J_\lambda,\omega_\lambda}}}.
\end{equation}
Finally, using again \eqref{eq:intersection of extensions} for the closed sets $G_-,G_+$ and using our assumption~\eqref{eqn:good regions_separation}, we deduce Property \ref{itm:Dobrushin-Shlosman_interface_edge-event_-,+_seminorm}.
\end{proof}
\end{proof}

\section{Phase co-existence in the box model}\label{sec:phase co-existence in box model}

In this section, we derive Theorems~\ref{thm:main} and \ref{thm:detail} as consequences of Theorem \ref{thm:phase-transition}.

Our convergence assumption describes four possible behaviors for $f_{\gamma}$ and the limiting $f$ (discrete vs.\ continuous and hard-core vs.\ soft-core). Our proofs will be streamlined to apply uniformly to all these cases, as much as possible. In the discrete case, for convenience in using the convergence assumption, we extend $f_{\gamma}$ from $S_\gamma=\gamma^d\Z\cap[0,\infty)$ to $[0,\infty)$ by a linear interpolation.

Throughout the section we fix $\beta>0$ and $0<\alpha<\alpha_{\max}$ such that the mean-field free energy $\phi_{\lambda}$ (see~\eqref{eq:mean-field free energy density}) is non-convex, as in the assumption of Theorem~\ref{thm:main}. The coupling constant $J_2>0$ and the dimension $d\ge 2$ are also held fixed.

\subsection{Deduction of Theorems~\ref{thm:main} and \ref{thm:detail}}
\label{sec:deduction of main theorem}

\paragraph{Translating to the $(J,\omega)$-spin model language} 

Let $0<\gamma\le 1$. 
To apply Theorem~\ref{thm:phase-transition}, we recall from Section \ref{sec:continuous family of spin models} the expression of the box model as a continuous family $(J_{\lambda,\gamma},\omega_{\lambda,\gamma})$ of $(J,\omega)$-spin models, indexed by the chemical potential $\lambda\in\R$:
\begin{align}
    J_{\lambda,\gamma}{}&=\frac{1}{2}J_2\beta\gamma^{-d},
    \\
    \dd{\omega}_{\lambda,\gamma}(\rho){}&=e^{-\beta\gamma^{-d}\phi_{\lambda,\gamma}(\rho)}\dd{\nu}_\gamma(\rho),
\end{align}
where the reference measure $\nu_\gamma$ was defined after~\eqref{eq:box model finite-volume Gibbs measure periodic} and we introduced the shorthand
\begin{equation}
    \phi_{\lambda,\gamma}(\rho):=-\lambda\rho-\frac{1}{2}\alpha\rho^2+f_{\gamma}(\rho).
\end{equation}
 
Our assumptions imply that $\omega_{\lambda,\gamma}(\R)\in(0,\infty)$ for sufficiently small $\gamma$ as required for $(J_{\lambda,\gamma},\omega_{\lambda,\gamma})$ to be a continuous family of $(J,\omega)$-spin models with respect to $\lambda$. 
The next lemma makes this precise, and also, in the soft-core case, clarifies that $f_{\gamma}$ must increase quadratically at infinity.
\begin{lemma}\label{lem:model is well defined}
    There exists $\gamma_0>0$ such that $\omega_{\lambda,\gamma}(\R)\in(0,\infty)$ for all $0<\gamma\le \gamma_0$ and $\lambda\in\R$. 
    In addition, for each $\alpha_0\in (0,\alpha_\max)$ there exists $\gamma_{\alpha_0}>0$ and $\rho_{\alpha_0}\in[0,\infty)$ such that $f_\gamma(\rho)\ge \frac{1}{2}\alpha_0\rho^2$ for $\rho\ge\rho_{\alpha_0}$ and $0<\gamma\le \gamma_{\alpha_0}$.
\end{lemma}
\begin{proof}
    The uniform convergence statements~\eqref{eq:uniform convergence below eta cp} (hard-core case) and~\eqref{eq:uniform convergence below infinity} (soft-core case) imply that if $\gamma$ is sufficiently small, then $f_\gamma<\infty$ on a subset of $S_\gamma$ of positive $\nu_\gamma$-measure, so that $\omega_{\lambda,\gamma}(\R)>0$. 

    In the hard-core case, $f_\gamma(\rho)=\infty$ for $\rho>\rho_\max$ so that the quadratic lower bound on $f_\gamma$ holds trivially. In the soft-core case, for each $\alpha_0\in (0,\alpha_\max)$, for the quadratic lower bound to fail there need to exist sequences $\gamma_n\downarrow 0$ and $\rho_n\uparrow\infty$ on which $f_{\gamma_n}(\rho_n)<\frac{1}{2}\alpha_0\rho_n^2$. However, as $\alpha_0<\alpha_\max$, this contradicts~\eqref{eq:growth at infinity}.

    The quadratic growth at infinity implies that $\int_{[\rho_1,\infty)} e^{-\beta\gamma^{-d}\phi_{\lambda,\gamma}(\rho)}\dd{\nu}_\gamma<\infty$ for some $\rho_1<\infty$ and all sufficiently small $\gamma$. To conclude that $\omega_{\lambda,\gamma}(\R)<\infty$ for small $\gamma$ and $\lambda\in\R$, it then suffices that $\inf_{\rho\in[0,\rho_1]}f_\gamma(\rho)>-\infty$ for small $\gamma$. In the soft-core case, this follows directly from~\eqref{eq:uniform convergence below infinity}. In the hard-core case, it follows from~\eqref{eq:uniform convergence below eta cp},~\eqref{eq:convergence at eta cp} and~\eqref{eq:uniform convergence above eta cp}.
\end{proof}

The family $(J_{\lambda,\gamma}, \omega_{\lambda,\gamma})$ is a continuous family of $(J,\omega)$-spin models with respect to $\lambda\in\R$, for any fixed values of $\beta,\gamma$, in the sense of Section \ref{sec:continuous family of spin models}. This follows from the continuity of $\phi_{\lambda,\gamma}$ in $\lambda$ and the dominated convergence theorem (using Lemma~\ref{lem:model is well defined}).

\paragraph{Applying Theorem~\ref{thm:phase-transition}}

In the hard-core case, set, in addition to \eqref{eq:f extension to rho cp},
\begin{equation}\label{eq:infinite value in extension}
     f(\rho):=\infty\quad\text{for $\rho>\rho_\cp$}.
 \end{equation}
With this extension, the domain of the function $\phi_\lambda$ also extends to $[0,\infty)$, for all $\lambda$, via its definition~\eqref{eq:mean-field free energy density}. We point out for later use that
\begin{equation}\label{eq:continuity of phi lambda}
    \text{$\phi_\lambda:[0,\infty)\to(-\infty,\infty]$ is continuous except possibly at $\rho_\cp$}.
\end{equation}
We proceed with the notation $\lambda_*, \rho_{*,-},\rho_{*,0},\rho_{*,+}$ as introduced before Theorem~\ref{thm:detail}. Define 
\begin{equation}                      m_*:=\min_{\rho\in[0,\infty)}\phi_{\lambda_*}(\rho),
\end{equation}
noting that it is necessarily finite.

\begin{remark}
    By Theorem \ref{thm:comparison_with_GP}, $-m_*$ is equal to the Gates--Penrose mean-field pressure.
\end{remark}

The following proposition verifies the assumption needed to apply Theorem~\ref{thm:phase-transition}.
\begin{proposition}
\label{prop:phase-transition}
    There exist $\kappa,\delta_0>0$ and $0<\gamma_0\le 1$ such that the following holds. Let $\delta\in(0,\delta_0)$.
    Define $\lambda_\pm(\delta) := \lambda_*\pm \kappa \delta$ and the closed sets
\begin{equation}\label{eq:G+- def}
\begin{split}
    G_-(\delta) &:= \phi_{\lambda_*}^{-1}([m_*, m_*+\delta])\cap [0,\rho_{*,0}),\\
    G_+(\delta) &:=\left(\phi_{\lambda_*}^{-1}([m_*, m_*+\delta])\cap (\rho_{*,0},+\infty)\right)\cup I(\delta),
\end{split}
\end{equation}
where $I(\delta):=\emptyset$ in either the soft-core case, or in the hard-core case when $\phi_{\lambda_*}(\rho_\cp)>m_*$, and where $I(\delta):=[\rho_\cp,\rho_\cp+\delta]$ in the hard-core case when $\phi_{\lambda_*}(\rho_\cp)=m_*$.

Then, for each $0<\gamma\le\gamma_0$, Assumption~\ref{as:good regions} is satisfied for the continuous family $(J_{\lambda,\gamma}, \omega_{\lambda,\gamma})_{\lambda\in[\lambda_-(\delta),\lambda_+(\delta)]}$ of $(J,\omega)$-spin models and the sets $G_\pm(\delta)$, with parameters $\theta_1(\delta,\gamma),\theta_2(\delta,\gamma),\theta_3(\delta,\gamma)$ which tend to zero as $\gamma\downarrow 0$.
    
\end{proposition}
\begin{remark}
The interval $I(\delta)$, added in the hard-core case when $\phi_{\lambda_*}(\rho_\cp)=m_*$, is required for the proposition to hold since it is possible, when $\phi_{\lambda_*}(\rho_\cp)=m_*$, that the measure $\tilde\omega_{\lambda,\gamma}$ concentrates its mass at densities $\rho>\rho_\cp$, which, in the absence of $I(\delta)$, would violate~\eqref{eqn:good regions_double-well} of Assumption~\ref{as:good regions}. We note, however, that the precise choice of $I(\delta)$ is unimportant and the proposition remains true (with the same proof) with $I(\delta)=[\rho_\cp, \rho_\cp+f(\delta)]$ for any $f(\delta)>0$.
\end{remark}

We now deduce Theorem~\ref{thm:detail} from Proposition \ref{prop:phase-transition}, noting that Theorem~\ref{thm:detail} immediately implies Theorem~\ref{thm:main}.

\begin{proof}[Proof of Theorem~\ref{thm:detail}]
    Let $\kappa,\delta_0>0$ and $0<\gamma_0\le 1$ be as in Proposition~\ref{prop:phase-transition}, and recall the notation $c(\epsilon,d)$ from Theorem~\ref{thm:phase-transition}.
    Fix two strictly decreasing sequences $(\delta_n)_{n\ge 1}\subset(0,\delta_0)$ and $(\epsilon_n)_{n\ge 1}\subset(0,1/2)$ such that $\lim_{n\to\infty}\delta_n=\lim_{n\to\infty}\epsilon_n=0$.
    By Proposition~\ref{prop:phase-transition}, for each $n\ge 1$, there exists $0<\gamma_n\le\gamma_0$ such that 
    \begin{equation}
    \label{eqn:detail-consolidation}
        \sup_{0<\gamma\le\gamma_n}\operatornamewithlimits{max}_{1\le i\le 3}\theta_i(\delta_n,\gamma)\le c(\epsilon_n,d).
    \end{equation}
    Without loss of generality, we may take $(\gamma_n)_{n\ge 1}$ to be strictly decreasing to $0$.

    We now specify a function $\lambda_c:(0,\gamma_1)\rightarrow\R$ satisfying the requirements of Theorem~\ref{thm:detail}, as follows.
    Since $(\gamma_n)_{n\ge 1}$ is strictly decreasing to $0$, for each $0<\gamma<\gamma_1$, there exists a unique $n(\gamma)\ge 1$ such that $\gamma\in[\gamma_{n(\gamma)+1},\gamma_{n(\gamma)})$.
    Using \eqref{eqn:detail-consolidation}, we let $\lambda_c(\gamma)$ be given by Theorem~\ref{thm:phase-transition}, i.e., such that $\lambda_c(\gamma)\in(\lambda_\ast-\kappa\delta_{n(\gamma)},\lambda_\ast+\kappa\delta_{n(\gamma)})$ and the $(J_{\lambda_c(\gamma),\gamma},\omega_{\lambda_c(\gamma),\gamma})$-spin model admits two distinct translation-invariant Gibbs measures $\P^{\pm}_{\lambda_c(\gamma),\gamma}$ satisfying
    \begin{equation}
        \label{eqn:detail-consolidation-density concentration}
        \P^{\pm}_{\lambda_c(\gamma),\gamma}(\eta_{0}\in G_{\pm}(\delta_{n(\gamma)}))\ge 1-\epsilon_{n(\gamma)}.
    \end{equation}
    
    We now verify that the function $\lambda_c$ chosen above satisfies the requirements of Theorem~\ref{thm:detail}.
    Property \eqref{eq:convergence of critical chemical potential} follows from the construction $\lambda_c(\gamma)\in(\lambda_\ast-\kappa\delta_{n(\gamma)},\lambda_\ast+\kappa\delta_{n(\gamma)})$, the monotone convergence of $(\delta_n)_{n\ge 1}$ to $0$, and the monotonicity of $n(\gamma)$ in $\gamma$.
    For Property \eqref{eq:density concentration}, let $U\subset\R$ be an open set containing $\mathcal{M}$.
    By the continuity of $\phi_{\lambda_\ast}$, there exists $\delta_U>0$ such that, for all $0<\delta\le\delta_U$, $G_-(\delta)\subset U\cap[0,\rho_{\ast,0})$ and $G_+(\delta)\subset U\cap(\rho_{\ast,0},\infty)$.
    By \eqref{eqn:detail-consolidation-density concentration}, for all small enough $\gamma>0$,
    \begin{equation}
        \label{eqn:detail-consolidation-density concentration-minus measure}
        \P^{-}_{\lambda_c(\gamma),\gamma}(\eta_{0}\in U\cap[0,\rho_{\ast,0}))
        \ge\P^{-}_{\lambda_c(\gamma),\gamma}(\eta_{0}\in G_{-}(\delta_{n(\gamma)}))
        \ge 1-\epsilon_{n(\gamma)},
    \end{equation}
    \begin{equation}
        \label{eqn:detail-consolidation-density concentration-plus measure}
        \P^{+}_{\lambda_c(\gamma),\gamma}(\eta_{0}\in U\cap(\rho_{\ast,0},\infty))
        \ge\P^{+}_{\lambda_c(\gamma),\gamma}(\eta_{0}\in G_{+}(\delta_{n(\gamma)}))
        \ge 1-\epsilon_{n(\gamma)}.
    \end{equation}
    Now, \eqref{eq:density concentration} follows by taking $\gamma\downarrow 0$ in \eqref{eqn:detail-consolidation-density concentration-minus measure} and \eqref{eqn:detail-consolidation-density concentration-plus measure}, and using the monotone convergence of $(\epsilon_n)_{n\ge 1}$ to $0$ and again the monotonicity of $n(\gamma)$ in $\gamma$.
\end{proof}

\subsection{Verifying Assumption~\ref{as:good regions}}

In this section, we deduce Proposition~\ref{prop:phase-transition} from the next proposition, which will itself be established in Section \ref{sec:proof of normalized measure estimate}.
The normalized measures $\tilde{\omega}_{\lambda,\gamma}$ are defined as in~\eqref{normalized_omega}.

\begin{proposition}
\label{prop:normalized-measure}
For all non-empty Borel $B\subseteq[0,\infty)$ such that $\inf_{\rho\in B} f(\rho)<\infty$ and all compact $K\subset\R$, 
\begin{equation}
\label{eqn:normalized-measure}
    \limsup_{\gamma\downarrow 0}\sup_{\lambda\in K}\left\{\gamma^d\log\tilde{\omega}_{\lambda,\gamma}(B)+\beta\left[\inf_{\rho\in \overline{B}}\phi_{\lambda}(\rho)-\inf_{\rho\in[0,\infty)}\phi_{\lambda}(\rho)\right]\right\}\le 0,
\end{equation}
with $\overline{B}$ denoting the closure of $B$.
\end{proposition}

\begin{lemma}\label{lem:terminal rho}
There exists $\rho_{\mathrm{T}}<\infty$ such that
\begin{equation}
    \inf_{\substack{\lambda\in[\lambda_*-1,\lambda_*+1]\\\rho\in[\rho_{\mathrm{T}},\infty)}} \phi_{\lambda}(\rho) \ge m_*+1.
\end{equation}
\end{lemma}

\begin{proof}
    In the hard-core case, $f(\rho)=\infty$ when $\rho>\rho_\cp$ by~\eqref{eq:infinite value in extension}, so we may take any $\rho_{\mathrm{T}}\in(\rho_\cp,\infty)$. 
    In the soft-core case, the claim follows from the quadratic lower bound~\eqref{eq:growth of f at infinity}, the fact that $\alpha<\alpha_\max$ and the definition~\eqref{eq:mean-field free energy density} of $\phi_{\lambda}$.
\end{proof}

\begin{proof}[Deduction of Proposition \ref{prop:phase-transition}]
Fix
\begin{equation}
    0<\delta_0<\begin{cases}\min\{\phi_{\lambda_*}(\rho_{*,0})-m_*, 1\}&\substack{\text{in the soft-core case,}\\\text{or the hard-core case with $\phi_{\lambda_*}(\rho_\cp)=m_*$}}\\
    \min\{\phi_{\lambda_*}(\rho_{*,0})-m_*, 1, \phi_{\lambda_*}(\rho_\cp)-m_*\}&\text{hard-core case with $\phi_{\lambda_*}(\rho_\cp)>m_*$}
    \end{cases}
\end{equation}
arbitrarily.
Fix the $0<\gamma_0\le 1$ of Lemma~\ref{lem:model is well defined}. Fix $0<\kappa<\min\{1, \frac{1}{\rho_{*,+}+\rho_{\mathrm{T}}}\}$ for the $\rho_{\mathrm{T}}$ of Lemma~\ref{lem:terminal rho}. 

Let $\delta\in(0,\delta_0)$. We first note that the sets $G_\pm(\delta)$ are closed. Recall that $\phi_{\lambda_*}$ is continuous in the soft-core case, and is continuous on $[0,\rho_\cp]$ (allowing it to be infinite at $\rho_\cp$) and is infinite on $(\rho_\cp,\infty)$ in the hard-core case.
Thus, as $m_*+\delta<\phi_{\lambda_*}(\rho_{*,0})$, we have that $G_-(\delta)= \phi_{\lambda_*}^{-1}([m_*, m_*+\delta])\cap [0,\rho_{*,0}-\eps]$ and $G_+(\delta) = (\phi_{\lambda_*}^{-1}([m_*, m_*+\delta])\cap [\rho_{*,0}+\eps,\infty))\cup I(\delta)$ for some $\eps>0$, whence $G_\pm(\delta)$ are closed (noting that $I(\delta)$ is closed).

Fix $K:=[\lambda_*-\kappa \delta, \lambda_*+\kappa\delta]$. Proposition~\ref{prop:normalized-measure} implies that for any set $B$ as in the proposition, there exists a function $\eps_{B}:(0,\gamma_0]\to[0,\infty)$ satisfying $\lim_{\gamma\downarrow0}\eps_{B}(\gamma)=0$ such that
\begin{equation}
\label{eqn:normalized measure with error}
    \tilde{\omega}_{\lambda,\gamma}(B)\le \exp{\gamma^{-d}\left(\eps_B(\gamma)-\beta\left[\inf_{\rho\in \overline{B}}\phi_{\lambda}(\rho)-\inf_{\rho\in[0,\infty)}\phi_{\lambda}(\rho)\right]\right)}\quad\text{for}\quad\substack{\lambda\in K,\\0<\gamma\le\gamma_0}.
\end{equation}

We start with~\eqref{eqn:good regions_separation} of Assumption~\ref{as:good regions}. We apply~\eqref{eqn:normalized measure with error} with $B=G_\pm(\delta)$, noting that $\inf_{\rho\in G_\pm(\delta)} f(\rho)<\infty$ by the definition of $G_\pm(\delta)$ and $m_*$. We obtain, for each $\lambda\in K$ and $0<\gamma\le\gamma_0$,
\begin{equation}
\begin{split}
    {}&e^{-\frac{1}{2}J_{\lambda,\gamma}\dist(G_-(\delta), G_+(\delta))^2}(\tilde\omega_{\lambda,\gamma}(G_-(\delta))\tilde\omega_{\lambda,\gamma}(G_+(\delta)))^{1/2}
    \\
    \le{}&
    \begin{multlined}[t]
        \exp\left\{-\frac{1}{2}J_{\lambda,\gamma}\dist(G_-(\delta), G_+(\delta))^2+\frac{1}{2}\gamma^{-d}\left(\epsilon_{G_-(\delta)}(\gamma)+\epsilon_{G_+(\delta)}(\gamma)
        \vphantom{\frac{1}{2}}\right.\right.
        \\
        \left.\left.-\beta\left[\inf_{\rho\in G_-(\delta)}\phi_{\lambda}(\rho)+\inf_{\rho\in G_+(\delta)}\phi_{\lambda}(\rho)-2\inf_{\rho\in [0,\infty)}\phi_{\lambda}(\rho))\right]\right)\right\}
    \end{multlined}
    \\
    \le{}&\exp{-\frac{1}{2}\gamma^{-d}\left[\frac{1}{2}J_2\beta\dist(G_-(\delta), G_+(\delta))^2-\epsilon_{G_-(\delta)}(\gamma)-\epsilon_{G_+(\delta)}(\gamma)\right]}.
\end{split}
\end{equation}
Since $G_-(\delta)$ and $G_+(\delta)$ are closed and disjoint we have that $\dist(G_-(\delta), G_+(\delta))>0$. Therefore,
\begin{equation}\label{eq:verifying the second property}
    \lim_{\gamma\downarrow 0}\sup_{\lambda\in K} e^{-\frac{1}{2}J_{\lambda,\gamma}\dist(G_-(\delta), G_+(\delta))^2}(\tilde\omega_{\lambda,\gamma}(G_-(\delta))\tilde\omega_{\lambda,\gamma}(G_+(\delta)))^{1/2} = 0.
\end{equation}

We continue with~\eqref{eqn:good regions_double-well} of Assumption~\ref{as:good regions}. We apply~\eqref{eqn:normalized measure with error} with $B=(G_-(\delta)\cup G_+(\delta))^{c}$, noting that $\inf_{\rho\in B} f(\rho)<\infty$ since $\rho_{*,0}\in B$ and, in the hard-core case, $\rho_{*,0}<\rho_{*,+}\le\rho_\cp$. We obtain, for each $\lambda\in K$ and $0<\gamma\le\gamma_0$,
\begin{equation}\label{eq:towards the first property}
\begin{multlined}
    \tilde\omega_{\lambda,\gamma}((G_-(\delta)\cup G_+(\delta))^{c})
    \\
    \le \exp{\gamma^{-d}\left(\eps_{(G_-(\delta)\cup G_+(\delta))^{c}}(\gamma)-\beta\left[\inf_{\rho\in \overline{(G_-(\delta)\cup G_+(\delta))^{c}}}\phi_{\lambda}(\rho)-\inf_{\rho\in[0,\infty)}\phi_{\lambda}(\rho)\right]\right)}.
\end{multlined}
\end{equation}
Now, for each $\lambda\in K$, on the one hand,
\begin{equation}
    \inf_{\rho\in[0,\infty)}\phi_{\lambda}(\rho)\le \phi_{\lambda}(\rho_{*,+})\le m_* + \kappa\delta\rho_{*,+}.
\end{equation}
On the other hand, recalling the definition of $G_\pm(\delta)$ from~\eqref{eq:G+- def} and applying Lemma~\ref{lem:terminal rho} using the fact that $K\subset[\lambda_*-1,\lambda_*+1]$ (since $\kappa,\delta<1$),
\begin{equation}\label{eq:infimum away from Gs}
    \inf_{\rho\in (G_-(\delta)\cup G_+(\delta))^{c}}\phi_{\lambda}(\rho)\ge m_* + \delta - \kappa\delta\rho_{\mathrm{T}} = m_* + (1 - \kappa\rho_{\mathrm{T}})\delta.
\end{equation}
Moreover, we claim that
\begin{equation}\label{eq:closure does not change inf}
    \inf_{\rho\in \overline{(G_-(\delta)\cup G_+(\delta))^{c}}}\phi_{\lambda}(\rho) = \inf_{\rho\in (G_-(\delta)\cup G_+(\delta))^{c}}\phi_{\lambda}(\rho).
\end{equation}
This is clear in the soft-core case since $\phi_\lambda$ is continuous. 
It follows in the hard-core case when $\phi_{\lambda_*}(\rho_\cp)>m_*$, since in this case $\rho_\cp$ belongs to the open set $(G_-(\delta)\cup G_+(\delta))^{c}$ by the definition of $G_\pm(\delta)$ from~\eqref{eq:G+- def} and our choice of $\delta$, whence the boundary $\partial (G_-(\delta)\cup G_+(\delta))$ consists only of continuity points of $\phi_\lambda$ by~\eqref{eq:continuity of phi lambda}. 
It also follows in the hard-core case when $\phi_{\lambda_*}(\rho_\cp)=m_*$ since a neighborhood of $\rho_\cp$ is contained in $G_+(\delta)$ by the definition~\eqref{eq:G+- def} (making use of $I(\delta)$), so again the boundary $\partial (G_-(\delta)\cup G_+(\delta))$ consists only of continuity points of $\phi_\lambda$ by~\eqref{eq:continuity of phi lambda}.
Therefore, since $\kappa<\frac{1}{\rho_{*,+}+\rho_{\mathrm{T}}}$,
\begin{equation}
    \inf_{\lambda\in K}\left(\inf_{\rho\in \overline{(G_-(\delta)\cup G_+(\delta))^{c}}}\phi_{\lambda}(\rho)-\inf_{\rho\in[0,\infty)}\phi_{\lambda}(\rho)\right)>0.
\end{equation}
Together with~\eqref{eq:towards the first property}, this implies
\begin{equation}\label{eq:verifying the first property}
    \lim_{\gamma\downarrow0}\sup_{\lambda\in K}\tilde\omega_{\lambda,\gamma}((G_-(\delta)\cup G_+(\delta))^{c}) = 0.
\end{equation}

Lastly, we check~\eqref{eqn:good regions_endpoints} of Assumption~\ref{as:good regions}.
Let $\#\in\set{-,+}$.
We apply~\eqref{eqn:normalized measure with error} with $B=G_\#(\delta)^{c}$, noting that $\inf_{\rho\in B} f(\rho)<\infty$ as, again, $\rho_{*,0}\in B$. We obtain, for each $0<\gamma\le\gamma_0$,
\begin{equation}\label{eq:towards the third property}
    \tilde\omega_{\lambda_\#,\gamma}(G_\#(\delta)^{c})\le \exp{\gamma^{-d}\left(\eps_{G_\#(\delta)^{c}}(\gamma)-\beta\left[\inf_{\rho\in \overline{G_\#(\delta)^{c}}}\phi_{\lambda_\#}(\rho)-\inf_{\rho\in[0,\infty)}\phi_{\lambda_\#}(\rho)\right]\right)}
\end{equation}
where we recall that $\lambda_\pm(\delta) = \lambda_*\pm \kappa \delta$. 
Now, on the one hand,
\begin{equation}
    \inf_{\rho\in G_{\mp}(\delta)}\phi_{\lambda_\pm}(\rho)
    \ge \inf_{\rho\in G_{\mp}(\delta)}\phi_{\lambda_\ast}(\rho)
    +\inf_{\rho\in G_{\mp}(\delta)}(-(\lambda_\pm-\lambda_\ast)\rho)
    \ge m_*-(\lambda_\pm-\lambda_*)\rho_{*,0},
\end{equation}
which, using that
\begin{equation}
    \overline{G_\pm(\delta)^c} 
    =G_\pm(\delta)^c\cup\partial G_\pm(\delta)
    = G_\mp(\delta)\cup \overline{(G_-(\delta)\cup G_+(\delta))^{c}}
\end{equation}
and using~\eqref{eq:closure does not change inf} and~\eqref{eq:infimum away from Gs}, implies that
\begin{equation}
\begin{multlined}
    \inf_{\rho\in \overline{G_\#(\delta)^{c}}}\phi_{\lambda_\#}(\rho)\ge m_*+\min\{-(\lambda_\#-\lambda_*)\rho_{*,0}, (1 - \kappa\rho_{\mathrm{T}})\delta\}
    \\
    =\begin{cases}m_*+\min\{\kappa\rho_{*,0},1 - \kappa\rho_{\mathrm{T}}\}\delta&\#=-\\
    m_*+\min\{-\kappa\rho_{*,0},1 - \kappa\rho_{\mathrm{T}}\}\delta&\#=+
    \end{cases}.
\end{multlined}
\end{equation}
On the other hand,
\begin{equation}
    \inf_{\rho\in[0,\infty)}\phi_{\lambda_\#}(\rho)\le \phi_{\lambda_\#}(\rho_{*,\#}) = m_* - (\lambda_\#-\lambda_*)\rho_{*,\#} = \begin{cases}m_*+\kappa\delta\rho_{*,-}&\#=-\\
    m_*-\kappa\delta\rho_{*,+}&\#=+
    \end{cases}.
\end{equation}
Therefore, using that $\rho_{*,-}<\rho_{*,0}<\rho_{*,+}$ and using again that $\kappa<\frac{1}{\rho_{*,+}+\rho_{\mathrm{T}}}$,
\begin{equation}
    \inf_{\rho\in \overline{G_\#(\delta)^{c}}}\phi_{\lambda_\#}(\rho)-\inf_{\rho\in[0,\infty)}\phi_{\lambda_\#}(\rho)>0.
\end{equation}
Together with~\eqref{eq:towards the third property}, this implies
\begin{equation}\label{eq:verifying the third property}
    \lim_{\gamma\downarrow0}\tilde\omega_{\lambda_\#,\gamma}(G_\#(\delta)^{c}) = 0.
\end{equation}

The proposition follows from~\eqref{eq:verifying the second property},~\eqref{eq:verifying the first property} and~\eqref{eq:verifying the third property}.
\end{proof}

\subsection{Technical lemmas}

We will deduce Proposition \ref{prop:normalized-measure} and Theorem \ref{thm:comparison_with_GP} from the two technical lemmas introduced in this section.
Recall that, in the discrete case, we extended $f_{\gamma}$ to $[0,\infty)$ by a linear interpolation. 

\begin{lemma}
\label{lem:normalized-measure_measure}
    For any compact $K\subset\R$ and non-empty, Borel $B\subseteq[0,\infty)$ such that $\inf_{\rho\in B} f(\rho)<\infty$,
    \begin{equation}
    \label{eqn:normalized-measure_measure}
    \limsup_{\gamma\downarrow 0}\sup_{\lambda\in K}\left\{\gamma^d\log\omega_{\lambda,\gamma}(B)+\beta\inf_{\rho\in \overline{B}}\phi_{\lambda}(\rho)\right\}\le 0
    \end{equation}    
\end{lemma}

\begin{claim}
\label{claim:box-model_hard-core_lower-bound}
For all non-empty, bounded Borel $B\subseteq[0,\infty)$ such that $\inf_{\rho\in B} f(\rho)<\infty$ and all compact $K\subset\R$,
\begin{equation}
    \liminf_{\gamma\downarrow0}\inf_{\lambda\in K}\left\{\inf_{\rho\in B}\phi_{\lambda,\gamma}(\rho)-\inf_{\rho\in \overline{B}}\phi_{\lambda}(\rho)\right\}\ge 0.
\end{equation}
\end{claim}

\begin{proof}
Suppose, to obtain a contradiction, that the claim does not hold. Therefore, there exist $\epsilon>0$ and sequences $(\gamma_j)\subset(0,1]$, $(\lambda_j)\subset K$, and $(\rho_j)\subset B$, with $\gamma_j\downarrow 0$, such that $\phi_{\lambda_j,\gamma_j}(\rho_j)<\infty$ and
\begin{equation}
    \phi_{\lambda_j,\gamma_j}(\rho_j)\le \inf_{\rho\in \overline{B}}\phi_{\lambda_j}(\rho)-\epsilon
\end{equation}
for all $j$.
By compactness, we may further assume that $\lambda_j\rightarrow\lambda_\infty\in K$ and $\rho_j\rightarrow\rho_\infty\in \overline{B}$.
To obtain a contradiction, it suffices to show that 
\begin{equation}\label{eq:limiting phi value}
    \liminf_{j\rightarrow\infty}\phi_{\lambda_j,\gamma_j}(\rho_j)\ge\phi_{\lambda_\infty}(\rho_\infty),
\end{equation}
using that $\lim_{j\rightarrow\infty}\inf_{\rho\in \overline{B}}\phi_{\lambda_j}(\rho)=\inf_{\rho\in \overline{B}}\phi_{\lambda_\infty}(\rho)$ by the boundedness of $B$ and the  continuity of $\phi_\lambda$.
In the soft-core case,~\eqref{eq:limiting phi value} follows from~\eqref{eq:uniform convergence below infinity}. In the hard-core case,~\eqref{eq:limiting phi value} follows from~\eqref{eq:uniform convergence below eta cp} if $\rho_\infty<\rho_\cp$, follows from~\eqref{eq:convergence at eta cp} (and~\eqref{eq:f extension to rho cp}) if $\rho_\infty=\rho_\cp$, and follows from~\eqref{eq:uniform convergence above eta cp} if $\rho_\infty>\rho_\cp$.
\end{proof}

\begin{claim}
\label{clm:tail bound}
For any compact $K\subset\R$, 
\begin{equation}
    \lim_{\rho_1\to\infty}\limsup_{\gamma\downarrow0}\sup_{\lambda\in K}\left\{\gamma^d\log\omega_{\lambda,\gamma}([\rho_1,\infty))\right\}=-\infty.
\end{equation}
\end{claim}
\begin{proof}
    Fix $\alpha_0\in(\alpha,\alpha_\max)$. 
    Lemma~\ref{lem:model is well defined} shows that there exist $\gamma_{\alpha_0}>0$ and $\rho_{\alpha_0}\in[0,\infty)$ such that $f_\gamma(\rho)\ge \frac{1}{2}\alpha_0\rho^2$ for $\rho\ge\rho_{\alpha_0}$ and $0<\gamma\le \gamma_{\alpha_0}$. 
    Therefore, for all $\lambda\in K$, and taking $0<\gamma\le \gamma_{\alpha_0}$ sufficiently small and $\rho_1\ge\rho_{\alpha_0}$ sufficiently large, it holds that
    \begin{multline}
        \int_{[\rho_1,\infty)}e^{-\beta\gamma^{-d}\phi_{\lambda,\gamma}(\rho)}\dd{\nu}_\gamma(\rho) \le \int_{[\rho_1,\infty)}e^{-\beta\gamma^{-d}(-\rho\min K + \frac{1}{2}(\alpha_0-\alpha)\rho^2)}\dd{\nu}_\gamma(\rho)\\
        \le\int_{[\rho_1,\infty)}e^{-\beta\gamma^{-d}(\rho+1)}\dd{\nu}_\gamma(\rho) \le e^{-\beta\gamma^{-d}\rho_1} \nu_\gamma([\rho_1,\rho_1+1))\sum_{k=1}^\infty e^{-\beta\gamma^{-d}k} \le e^{-\beta\gamma^{-d}\rho_1},
    \end{multline}
    where the second inequality uses that $-\rho\min K + \frac{1}{2}(\alpha_0-\alpha)\rho^2\ge \rho+1$ for sufficiently large $\rho$, the third inequality uses the monotonicity of the integrand and the $1$-periodicity of $\nu_\gamma$, and the final inequality uses that $\gamma$ is sufficiently small and the definition of $\nu_\gamma$. 
    The claim follows.
\end{proof}

\begin{proof}[Proof of Lemma \ref{lem:normalized-measure_measure}]
Let $K\subset\R$ be compact and $B\subseteq[0,\infty)$ be non-empty and Borel, satisfying that $\inf_{\rho\in B} f(\rho)<\infty$.
To prove \eqref{eqn:normalized-measure_measure}, we first use Claim \ref{clm:tail bound} to find $\rho_1>0$ such that
\begin{equation}
\label{eqn:normalized-measure_rho1-second-property}
    \limsup_{\gamma\downarrow 0}\sup_{\lambda\in K}\left\{\gamma^d\log\int_{[\rho_1,\infty)}e^{-\beta\gamma^{-d}\phi_{\lambda,\gamma}(\rho)}\dd{\nu}_\gamma(\rho)\right\}
    \le -\beta\sup_{\lambda\in K}\inf_{\rho\in\overline{B}}\phi_\lambda(\rho).
\end{equation}
Let $B_1:=B\cap[0,\rho_1]$ and $B_2:=B\cap(\rho_1,\infty)$.
Splitting $\omega_{\lambda,\gamma}(B)=\omega_{\lambda,\gamma}(B_1)+\omega_{\lambda,\gamma}(B_2)$ and using the elementary inequality $\log(a+b)\le\log2+\max\set{\log a,\log b}$, we bound the LHS of \eqref{eqn:normalized-measure_measure} by
\begin{equation}
\label{eqn:normalized-measure_measure_max-bound}
\begin{multlined}
    \limsup_{\gamma\downarrow 0}\sup_{\lambda\in K}\left\{\gamma^d\log2+\max\set{
    \log\omega_{\lambda,\gamma}(B_1),
    \log\omega_{\lambda,\gamma}(B_2)
    }+\beta\inf_{\rho\in \overline{B}}\phi_{\lambda}(\rho)\right\}
    \\
    =\max
    \left\{\limsup_{\gamma\downarrow 0}\sup_{\lambda\in K}\left\{\gamma^d
    \log\omega_{\lambda,\gamma}(B_1)
    +\beta\inf_{\rho\in \overline{B}}\phi_{\lambda}(\rho)\right\},
    \right.
    \\
    \left.\limsup_{\gamma\downarrow 0}\sup_{\lambda\in K}\left\{\gamma^d
    \log\omega_{\lambda,\gamma}(B_2)
    +\beta\inf_{\rho\in \overline{B}}\phi_{\lambda}(\rho)\right\}
    \right\}.
\end{multlined}
\end{equation}
On the one hand, using Claim \ref{claim:box-model_hard-core_lower-bound} and that $\inf_{\rho\in B_1}\phi_{\lambda,\gamma}(\rho)\ge \inf_{\rho\in B}\phi_{\lambda,\gamma}(\rho)$,
\begin{equation}
\label{eqn:normalized-measure_measure_main-part}
\begin{multlined}
    \limsup_{\gamma\downarrow 0}\sup_{\lambda\in K}\left\{\gamma^d
    \log\omega_{\lambda,\gamma}(B_1)
    +\beta\inf_{\rho\in \overline{B}}\phi_{\lambda}(\rho)\right\}
    \\
    \le \limsup_{\gamma\downarrow 0}\gamma^d\log\nu_\gamma(B_1)
    -\beta\liminf_{\gamma\downarrow 0}\inf_{\lambda\in K}\left\{\inf_{\rho\in B}\phi_{\lambda,\gamma}(\rho)-\inf_{\rho\in \overline{B}}\phi_{\lambda}(\rho)\right\}
    \le 0.
\end{multlined}
\end{equation}
On the other hand, using \eqref{eqn:normalized-measure_rho1-second-property},
\begin{equation}
\label{eqn:normalized-measure_measure_tail}
\begin{multlined}
    \limsup_{\gamma\downarrow 0}\sup_{\lambda\in K}\left\{\gamma^d
    \log\omega_{\lambda,\gamma}(B_2)
    +\beta\inf_{\rho\in \overline{B}}\phi_{\lambda}(\rho)\right\}
    \\
    \le\limsup_{\gamma\downarrow 0}\sup_{\lambda\in K}\left\{\gamma^d\log\int_{[\rho_1,\infty)} e^{-\beta\gamma^{-d}\phi_{\lambda,\gamma}(\rho)}\dd{\nu}_\gamma(\rho)
    \right\}+\beta\sup_{\lambda\in K}\inf_{\rho\in \overline{B}}\phi_{\lambda}(\rho)
    \le 0.
\end{multlined}
\end{equation}
Combining \eqref{eqn:normalized-measure_measure_max-bound}, \eqref{eqn:normalized-measure_measure_main-part}, and \eqref{eqn:normalized-measure_measure_tail}, we get \eqref{eqn:normalized-measure_measure}.
\end{proof}

\begin{lemma}
    \label{lem:normalized-measure_normalization}
    For any compact $K\subset\R$,
    \begin{equation}
    \label{eqn:normalized-measure_normalization}
    \liminf_{\gamma\downarrow 0}\inf_{\lambda\in K}\inf_{L\ge 1}\left\{\frac{1}{\gamma^{-d}\abs{\Lambda_L}}\log\Xi^{L,\per}_{\lambda,\gamma}+\beta\inf_{\rho\in[0,\infty)}\phi_\lambda(\rho)\right\}\ge 0,
    \end{equation}
    where we introduced the shorthand $\psi_{\lambda,\gamma}:=\psi_{J_{\lambda,\gamma},\omega_{\lambda,\gamma}}$.
\end{lemma}

\begin{claim}
\label{claim:box-model_pressure-localization}
For any compact $K\subset\R$,
\begin{equation}
    \lim_{\xi\downarrow0}\limsup_{\gamma\downarrow0}\sup_{\lambda\in K}\inf_{\rho_0\in[0,\infty)}\left\{\sup_{\rho\in[\rho_0,\rho_0+\xi]}\phi_{\lambda,\gamma}(\rho)-\inf_{\rho\in[0,\infty)}\phi_\lambda(\rho)\right\}\le 0.
\end{equation}
\end{claim}

\begin{proof}
    Let $\epsilon>0$. In the soft-core case, choose $\rho_1<\infty$ such that
    \begin{equation}
        \inf_{\rho\in[0,\rho_1]}\phi_\lambda(\rho)=\inf_{\rho\in[0,\infty)}\phi_\lambda(\rho)\quad\text{for all $\lambda\in K$.}
    \end{equation}
    This is possible since $K$ is bounded and using the quadratic growth~\eqref{eq:growth of f at infinity} of $f$, together with our choice $\alpha<\alpha_\max$ (and the definition~\eqref{eq:mean-field free energy density} of $\phi_\lambda$). In the hard-core case, choose $\rho_1<\rho_\cp$ such that
    \begin{equation}
        \sup_{\lambda\in K}\left\{\inf_{\rho\in[0,\rho_1]}\phi_\lambda(\rho)-\inf_{\rho\in[0,\rho_\cp]}\phi_\lambda(\rho)\right\}\le\epsilon.
    \end{equation}
    This is possible since $K$ is bounded and as $f$ is continuous on $[0,\rho_\cp]$ (at $\rho_\cp$, we mean this in the generalized sense~\eqref{eq:f extension to rho cp} if $f(\rho_\cp)=\infty$).
        
    As $\phi_\lambda$ is continuous, for each $\lambda\in K$, there exists $\rho_0(\lambda)\in[0,\rho_1]$ such that $\phi_\lambda(\rho_0(\lambda))=\inf_{\rho\in[0,\rho_1]}\phi_\lambda(\rho)$.
    Then, for all small enough $\xi,\gamma>0$,
    \begin{equation}
        \begin{split}
            {}&\sup_{\lambda\in K}\inf_{\rho_0\in[0,\infty)}\left\{\sup_{\rho\in[\rho_0,\rho_0+\xi]}\phi_{\lambda,\gamma}(\rho)-\inf_{\rho\in[0,\infty)}\phi_\lambda(\rho)\right\}
            \\
            \le{}&\sup_{\lambda\in K}\left\{\sup_{\rho\in[\rho_0(\lambda),\rho_0(\lambda)+\xi]}\phi_{\lambda,\gamma}(\rho)-\inf_{\rho\in[0,\infty)}\phi_\lambda(\rho)\right\}
            \\
            \le{}&\sup_{\lambda\in K}\left\{\sup_{\rho\in[\rho_0(\lambda),\rho_0(\lambda)+\xi]}\phi_{\lambda}(\rho)-\inf_{\rho\in[0,\infty)}\phi_\lambda(\rho)\right\}+\epsilon
            \\
            \le{}&\sup_{\lambda\in K}\left\{\phi_{\lambda}(\rho_0(\lambda))-\inf_{\rho\in[0,\infty)}\phi_\lambda(\rho)\right\}+2\epsilon
            \\
            \le{}&3\epsilon,
        \end{split}
    \end{equation}
    where we used the uniform convergence assumption~\eqref{eq:uniform convergence below eta cp} (hard-core case) or~\eqref{eq:uniform convergence below infinity} (soft-core case) in the second inequality, the uniform continuity of $(\lambda,\rho)\mapsto\phi_\lambda(\rho)$ on $K\times[0,\rho_1+\xi]$ in the third, and the definition of $\rho_0(\lambda)$ and $\rho_1$ in the last.
    The proof is complete after taking $\gamma\downarrow 0$, $\xi\downarrow 0$, and $\epsilon\downarrow 0$.
\end{proof}

\begin{proof}[Proof of Lemma \ref{lem:normalized-measure_normalization}]
To prove~\eqref{eqn:normalized-measure_normalization}, we bound, using \eqref{eqn:partition-function-lower-bound},
\begin{equation}
    \inf_{L\ge 1}\frac{1}{\gamma^{-d}\abs{\Lambda_L}}\log\Xi^{L,\per}_{\lambda,\gamma}
    \ge \sup_{S}\left\{-dJ_{\lambda,\gamma}\diam(S)^2+\log\omega_{\lambda,\gamma}(S)\right\}.
\end{equation}
Let $\xi>0$. 
By restricting to sets $S$ of the form $[\rho_0,\rho_0+\xi]$, where $\rho_0\in [0,\infty)$, we bound the LHS of \eqref{eqn:normalized-measure_normalization} below by
\begin{equation}
\begin{multlined}
    \liminf_{\gamma\downarrow 0}\inf_{\lambda\in K}\sup_{\rho_0\in[0,\infty)}\left\{-\gamma^d dJ_{\lambda,\gamma}\xi^2+\gamma^d\log\omega_{\lambda,\gamma}(S)+\beta\inf_{\rho\in[0,\infty)}\phi_\lambda(\rho)\right\}
    \\
    =-\frac{1}{2}\beta d J_2\xi^2
    +\liminf_{\gamma\downarrow 0}\inf_{\lambda\in K}\sup_{\rho_0\in[0,\infty)}\left\{\gamma^d\log\omega_{\lambda,\gamma}(S)+\beta\inf_{\rho\in[0,\infty)}\phi_\lambda(\rho)\right\},
\end{multlined}
\end{equation}
where $S$ is the shorthand for $[\rho_0,\rho_0+\xi]$.
Thus, it suffices to show that
\begin{equation}
    \lim_{\xi\downarrow 0}\liminf_{\gamma\downarrow 0}\inf_{\lambda\in K}\sup_{\rho_0\in[0,\infty)}\left\{\gamma^d\log\int_{[\rho_0,\rho_0+\xi]}e^{-\beta\gamma^{-d}\phi_{\lambda,\gamma}(\rho)}\dd{\nu}_\gamma(\rho)+\beta\inf_{\rho\in[0,\infty)}\phi_\lambda(\rho)\right\}\ge 0.
\end{equation}
We bound
\begin{equation}
\begin{split}
    {}&\sup_{\rho_0\in[0,\infty)}\left\{\gamma^d\log\int_{[\rho_0,\rho_0+\xi]}e^{-\beta\gamma^{-d}\phi_{\lambda,\gamma}(\rho)}\dd{\nu}_\gamma(\rho)+\beta\inf_{\rho\in[0,\infty)}\phi_\lambda(\rho)\right\}
    \\
    \ge{}&\sup_{\rho_0\in[0,\infty)}\left\{\gamma^d\log\nu_\gamma([\rho_0,\rho_0+\xi])-\beta\sup_{\rho\in[\rho_0,\rho_0+\xi]}\phi_{\lambda,\gamma}(\rho)+\beta\inf_{\rho\in[0,\infty)}\phi_\lambda(\rho)\right\}
    \\
    \ge{}&\inf_{\rho_0\in[0,\infty)}\left\{\gamma^d\log\nu_\gamma([\rho_0,\rho_0+\xi])\right\}
    -\beta\inf_{\rho_0\in[0,\infty)}\left\{\sup_{\rho\in[\rho_0,\rho_0+\xi]}\phi_{\lambda,\gamma}(\rho)-\inf_{\rho\in[0,\infty)}\phi_\lambda(\rho)\right\}.
\end{split}
\end{equation}
Since
\begin{equation}
    \liminf_{\gamma\downarrow 0}\inf_{\rho_0\in[0,\infty)}\left\{\gamma^d\log\nu_\gamma([\rho_0,\rho_0+\xi])\right\}=0,
\end{equation}
we deduce~\eqref{eqn:normalized-measure_normalization} using Claim~\ref{claim:box-model_pressure-localization}.
\end{proof}

\subsection{Deduction of Proposition~\ref{prop:normalized-measure} from Lemma~\ref{lem:normalized-measure_measure} and Lemma~\ref{lem:normalized-measure_normalization}}
\label{sec:proof of normalized measure estimate}

Let $K\subset\R$ be compact and $B\subseteq[0,\infty)$ be non-empty and Borel, satisfying that $\inf_{\rho\in B} f(\rho)<\infty$.
By \eqref{normalized_omega},
\begin{equation}
    \gamma^d\log\tilde{\omega}_{\lambda,\gamma}(B)
    =\gamma^d\log\omega_{\lambda,\gamma}(B)
    -\gamma^d\psi_{\lambda,\gamma},
\end{equation}
so the LHS of \eqref{eqn:normalized-measure} is bounded above by
\begin{equation}
\begin{multlined}
    \limsup_{\gamma\downarrow 0}\sup_{\lambda\in K}\left\{\gamma^d\log\omega_{\lambda,\gamma}(B)+\beta\inf_{\rho\in \overline{B}}\phi_{\lambda}(\rho)\right\}
    \\
    +\limsup_{\gamma\downarrow 0}\sup_{\lambda\in K}\left\{-\gamma^d\psi_{\lambda,\gamma}-\beta\inf_{\rho\in[0,\infty)}\phi_\lambda(\rho)\right\}.
\end{multlined}
\end{equation}
The proposition now follows from Lemmas \ref{lem:normalized-measure_measure} and \ref{lem:normalized-measure_normalization} and the definition \eqref{eqn:spin-model_free-energy} of $\psi_{\lambda,\gamma}$.

\subsection{Proof of Theorem~\ref{thm:comparison_with_GP}}

Let $\lambda\in\R$.
By Lemma~\ref{lem:normalized-measure_normalization},
\begin{equation}
\label{eqn:comparison_with_GP-lower_bound}
    \liminf_{\substack{L\to\infty \\ \gamma\downarrow 0}}\left\{\frac{1}{\gamma^{-d}|\Lambda_L|}\log\Xi^{L, \per}_{\lambda,\gamma}+\beta\inf_{\rho}\phi_\lambda(\rho)\right\}\ge 0.
\end{equation}
For an upper bound, we use the trivial bound ${H^L_{J_{\lambda,\gamma}}(\eta)}\ge 0$ in \eqref{partition_function} to obtain
\begin{equation}
    \sup_{L\ge 1}\left\{\frac{1}{\gamma^{-d}|\Lambda_L|}\log\Xi^{L, \per}_{\lambda,\gamma}\right\}\le\log\omega_{\lambda,\gamma}([0,\infty)).
\end{equation}
Applying Lemma~\ref{lem:normalized-measure_measure} with $B=[0,\infty)$, we get that
\begin{equation}
\label{eqn:comparison_with_GP-upper_bound}
    \limsup_{\substack{L\to\infty\\\gamma\downarrow 0}}\left\{\frac{1}{\gamma^{-d}|\Lambda_L|}\log\Xi^{L, \per}_{\lambda,\gamma}+\beta\inf_{\rho}\phi_{\lambda}(\rho)\right\}
    \le\limsup_{\gamma\downarrow 0}\left\{\gamma^d\log\omega_{\lambda,\gamma}(B)+\beta\inf_{\rho}\phi_{\lambda}(\rho)\right\}
    \le 0.
\end{equation}
The theorem follows from \eqref{eqn:comparison_with_GP-lower_bound} and \eqref{eqn:comparison_with_GP-upper_bound}.

\subsection*{Acknowledgements}
JLL thanks Roman Koteck\'y for emphasizing the importance of the problem of rigorously establishing the liquid-vapor phase transition in a 2022 IHES meeting.

QH is supported by an SAS fellowship at Rutgers University.
The research of RP is partially supported by the Israel Science Foundation grants 1971/19 and 2340/23, by the European Research Council Consolidator grant 101002733 (Transitions) and by the National Science Foundation grant DMS-2451133.
Part of this work was completed while RP was a visiting fellow at the Mathematics Department of Princeton University, a visitor of the Institute for Advanced Study and a consultant at Rutgers University. RP is grateful for their support.
IJ gratefully acknowledges support through NSF Grant DMS-2349077, and the Simons Foundation, Grant Number 825876.

\bibliographystyle{plain}
\bibliography{bibliography}

\appendix

\section{Convergence assumptions for particle systems}\label{app:ruelle}

The convergence assumptions in Section \ref{sec:convergence_assumptions} are satisfied for systems of particles interacting via {\it well-behaved} pair potentials.
These results were proved by Ruelle \cite{ruelle1963classical, Ruelle69}, and are recalled in this appendix.

Consider a continuum particle system in the box $[0,\gamma^{-1}]^d$, interacting via the Hamiltonian
\begin{equation}
  H(x_1,\cdots,x_N)=\sum_{i<j}\phi(x_i-x_j)
\end{equation}
where $\phi$ is an even function of $\mathbb R^d$.
We say that $\phi$ is {\it stable} \cite[Definition 3.2.1]{Ruelle69} if there exists $B \geqslant 0$ such that
\begin{equation}
  H(x_1,\cdots,x_N) \geqslant -NB
\end{equation}
and {\it tempered} \cite[(1.12)]{Ruelle69} if
\begin{equation}
  \phi(x) \leqslant A|x|^{-\lambda}
  \quad \mathrm{for}\quad
  |x| \geqslant R_0
\end{equation}
for some $\lambda>d$ and $A,R_0>0$.
Let $f_\gamma(\rho)$ denote the canonical free energy
\begin{equation}
  f_\gamma(N \gamma^{d}):=-\frac1{\beta \gamma^{-d}}\log \frac1{N!}\int \dd{x}_1\cdots \dd{x}_Ne^{-\beta H(x_1,\cdots,x_N)}
  \label{fparticle}
\end{equation}
($\beta$ is the inverse temperature, which we consider fixed, and so do not make it explicit in the notation.)

\begin{theorem}{\rm(\cite[Theorem 3.4.4]{Ruelle69})}:\label{theo:particle}
  If the potential is stable and tempered, then there exists $\rho_{\mathrm{cp}}\in[0,\infty]$ and a convex (and thus continuous) function $f:[0,\rho_{\mathrm{cp}})\to \mathbb R$ such that, for all $\rho \geqslant 0$ and any $\rho_\gamma$ such that $\rho_\gamma\to\rho$, the following hold.
  If $\rho<\rho_{\mathrm{cp}}$
  \begin{equation}
    \lim_{\gamma\downarrow 0}f_\gamma(\rho_\gamma)=f(\rho)
    \label{ruelle_lim1}
  \end{equation}
  if $\rho=\rho_{\mathrm{cp}}$
  \begin{equation}
    \liminf_{\gamma\downarrow 0}f_\gamma(\rho_\gamma)\geqslant \lim_{\rho\uparrow\rho_{\mathrm{cp}}}f(\rho)
    \label{ruelle_lim2}
  \end{equation}
  and if $\rho>\rho_{\mathrm{cp}}$ then
  \begin{equation}
    \lim_{\gamma\downarrow0}f_\gamma(\rho_\gamma)=\infty
    .
    \label{ruelle_lim3}
  \end{equation}
\end{theorem}

The limit in (\ref{ruelle_lim1}) is actually uniform, as mentioned (in a slightly different context) in \cite[Remark 3.3.13]{Ruelle69}.
We give the argument here for the sake of completeness.

\begin{corollary}\label{cor:uniform}
  The limit in (\ref{ruelle_lim1}) is uniform.
\end{corollary}

\begin{proof}
  Suppose the limit were not uniform, then there would be $\epsilon>0$ and a sequence $\gamma_i\to0$ such that $|f_{\gamma_i}(n_{\gamma_i})-f(\rho)| \geqslant \epsilon$.
  This contradicts Theorem \ref{theo:particle} since it applies to {\it any} $\rho_\gamma$ and in particular to $n_{\gamma_i}$.
\end{proof}

As we will see below, these results allow us to prove  assumptions 1(a), 1(b), 1(c), and 2(a) of Section \ref{sec:convergence_assumptions}.
To obtain assumption 2(b), we will need to impose a stronger stability condition.
The potential is said to be {\it superstable} if \cite[Section 3.2.9]{Ruelle69}
\begin{equation}
  H(x_1,\cdots,x_N) \geqslant N(C N \gamma^{d}-D)
  \label{superstability}
\end{equation}
for some $C>0$ and $D \geqslant 0$.
(Note that superstability trivially implies stability.)

\begin{lemma}\label{lemma:superstable}
  If the potential is superstable, then
  \begin{equation}
    f_\gamma(\rho)
    \geqslant \rho\left(\rho C-D
    +\frac 1{\beta}(\log\rho-1)
    \right)
    .
  \end{equation}
\end{lemma}

\begin{proof}
  Plugging (\ref{superstability}) into (\ref{fparticle}), we find
  \begin{equation}
    f_\gamma(N \gamma^{d})
    \geqslant -\frac1{\beta \gamma^{-d}}\log\frac{\gamma^{-Nd}}{N!}e^{-\beta N^2\gamma^{d}C+\beta ND}
  \end{equation}
  and we conclude using $N! \geqslant N^N e^{-N}$.
\end{proof}

Superstability will thus allow us to ensure assumption 2(b) of Section \ref{sec:convergence_assumptions}.
However, the condition (\ref{superstability}) is not very explicit.
Ruelle derived \cite{ruelle1963classical} a more elementary condition on the potential that implies superstability.
Since in \cite{ruelle1963classical}, this condition is formulated only in three dimensions, and is not separated clearly from the rest of the discussion, we state Ruelle's result here and give a proof, following Ruelle's original.

Ruelle proved \cite{ruelle1963classical} that, if $\phi$ is superstable, then $f_\gamma$ grows at least quadratically at infinity.
This result is not written explicitly in \cite{ruelle1963classical} as a theorem, so we reproduce its proof here.

\begin{lemma}{\rm(\cite{ruelle1963classical})}\label{lemma:super}
  If $\phi(x) \geqslant \phi_0(x)$ where $\phi_0$ is continuous, Lebesgue-integrable, $\int \phi_0(x)\dd{x}>0$ and the Fourier transform of $\phi_0$ is non-negative, then $\phi$ is superstable for sufficiently small $\gamma$.
\end{lemma}

\begin{proof}
  We bound
  \begin{equation}
    H(x_1,\cdots,x_N) \geqslant
    \sum_{i<j}\phi_0(x_i-x_j)
    =
    \frac12\sum_{i,j}\phi_0(x_i-x_j)-\frac N2\phi_0(0)
    .
  \end{equation}
  Now, the Fourier transform of $\phi_0$ is defined as
  \begin{equation}
    \hat\phi_0(k):=\frac1{(2\pi)^{\frac d2}}\int \dd{x}e^{ikx}\phi_0(x)
  \end{equation}
  in terms of which
  \begin{equation}
    \sum_{i,j}\phi_0(x_i-x_j)=
    \frac1{(2\pi)^{\frac d2}}\int \dd{k}\hat\phi_0(k)\sum_{i,j}e^{i(x_i-x_j)k}
    .
    \label{phi2_fourier}
  \end{equation}
  Now, given any unit vector $u$, $\sum_{i,j}e^{ip(x_i-x_j)u}$ is entire in $p$, and it is real, so, expanding the exponential will only yield the even terms:
  \begin{equation}
    \sum_{i,j}e^{ip(x_i-x_j)u}
    =\sum_{n=0}^\infty \frac{(-1)^np^{2n}}{2n!}\sum_{i,j}((x_i-x_j)u)^{2n}
  \end{equation}
  and since $|x_i-x_j|\leqslant \gamma^{-1}\sqrt d$,
  \begin{equation}
    \begin{array}{>\displaystyle l}
    \sum_{i,j}e^{ip(x_i-x_j)u}
    \geqslant\sum_{i,j}\left(1-\sum_{n=1}^\infty \frac{p^{2(2n-1)}}{(2(2n-1))!}(\gamma^{-1}\sqrt d)^{2(2n-1)}\right)
    =\\[0.5cm]\hfill
    =N^2\left(1-\frac{\cosh(\gamma^{-1}p\sqrt d)-\cos(\gamma^{-1}p\sqrt d)}2\right)
    .
  \end{array}
  \end{equation}
  Defining
  \begin{equation}
    f(p):=\max\left\{0,\left(1-\frac{\cosh(p\sqrt d)-\cos(p\sqrt d)}2\right)\right\}
  \end{equation}
  we thus have
  \begin{equation}
    \sum_{i,j}e^{ip(x_i-x_j)u}
    \geqslant
    N^2f(\gamma^{-1}p)
    .
  \end{equation}
  Plugging this into (\ref{phi2_fourier}) we have, recalling that $\hat\phi_0 \geqslant 0$,
  \begin{equation}
    \sum_{i,j}\phi_0(x_i-x_j)\geqslant
    \frac{N^2}{(2\pi)^{\frac d2}}\int \dd{k}\hat\phi_0(k)f(\gamma^{-1}|k|)
    =
    \frac{N^2}{\gamma^d(2\pi)^{\frac d2}}\int \dd{k}\hat\phi_0(\gamma k)f(|k|)
    .
  \end{equation}
  Note that $f$ has compact support and $\hat\phi_0$ is bounded (since $\phi_0$ is integrable) so, by the dominated convergence theorem,
  \begin{equation}
    \lim_{\gamma\to0}\int \dd{k}\hat\phi_0(\gamma k)f(|k|)
    =\hat\phi_0(0)\int \dd{k}f(|k|)
  \end{equation}
  and so, for any $A <\hat\phi_0(0)\int \dd{k}f(|k|)$, there exists $\gamma_0$ such that, if $\gamma \geqslant \gamma_0$, then
  \begin{equation}
    \int \dd{k}\hat\phi_0(\gamma k)f(|k|) \geqslant A
    .
  \end{equation}
  Putting all this together, we find that
  \begin{equation}
    H(x_1,\cdots,x_N)\geqslant \frac{N^2}{2 \gamma^{-d}(2\pi)^{\frac d2}}A-\frac N2\phi_0(0)
    .
  \end{equation}
  
  In summary, $H$ is superstable with $D\equiv \max\{0,\frac12\phi_0(0)\}$ and $C$ can be chosen arbitrarily close to $\frac12(2\pi)^{-\frac d2}\hat\phi_0(0)\int \dd{k}f(|k|)$ (at the expense of making $\gamma$ smaller).
\end{proof}

We now have all the ingredients to easily prove the following proposition.

\begin{proposition}\label{prop:convergence_assumptions}
  If the potential is superstable and tempered, then the conditions of Section \ref{sec:convergence_assumptions} are satisfied.
\end{proposition}

\begin{proof}
  Conditions 1(a), and 2(a) follow immediately from (\ref{ruelle_lim1}) and Corollary \ref{cor:uniform}.
  Condition 1(b) follows from (\ref{ruelle_lim2}).
  Condition 2(b) an immediate consequence of Lemma \ref{lemma:superstable}.
  We are left with condition 1(c).

  Suppose the potential has a hard-core: $\phi(x)=\infty$ for $|x|<R$.
  The existence of $\rho_{\mathrm{max}}$ is then obvious.
  Now, $f_\gamma$ is continuous, as it is obtained from the discrete function $f(N \gamma^{d})$ using a linear interpolation, so the infimum of $f_\gamma(\rho)$ over $\rho\in[\rho_1,\rho_{\mathrm{max}}]$ is reached, say at $\rho_\gamma$.
  Condition 1(c) then follows from (\ref{ruelle_lim3}).
\end{proof}

Finally, using Lemma \ref{lemma:super}, we can find many examples of superstable, tempered potentials, and thus find particle models that satisfy the conditions of Section \ref{sec:convergence_assumptions}.

\begin{proposition}{\rm\cite[Appendix]{ruelle1963classical}}\label{prop:superstable_examples}
  The following potentials are superstable and tempered:
  \begin{itemize}
    \item
    In any dimension,
    $\phi(x)\geqslant 0$, $\phi$ is compactly supported, and $\phi(x) \geqslant a>0$ in the vicinity of the origin.

    \item
    In three dimensions, the
    Lennard--Jones potential: $\phi(x)=4 \epsilon((R/|x|)^{12}-(R/|x|)^6)$.

    \item
    In three dimensions, the
    Morse potential: $\phi(x)=\epsilon(e^{-2 \alpha(|x|-R)}-2e^{-\alpha(|x|-R)})$ for $e^{\alpha R}>16$.
  \end{itemize}
\end{proposition}

It is proved in \cite{ruelle1963classical} that these potentials are superstable (the first bullet point is only stated in three dimensions, but its proof extends trivially to arbitrary dimensions).
The fact that they are tempered is obvious.

\section{Continuity results}

We prove here a useful continuity property of Gibbs measures of the $(J,\omega)$-spin models, used in the proof of Theorem \ref{thm:DS theorem}, namely that any (subsequential) limit of these measures in distribution along a convergent sequence $((J_j,\omega_j))_{j\ge 1}$, $(J_j,\omega_j)\rightarrow(J,\omega)$, is a \emph{Gibbs measure} of the $(J,\omega)$-spin model.
We note that a result of a similar flavor is proven in \cite[Theorem 4.17]{georgii2011gibbs}, although it does not apply directly to our situation.

\begin{proposition}
\label{prop:convergence-of-Gibbs-measures}
Let $J_j,J\ge 0$, and $\omega_j,\omega$, $j\ge 1$, be Borel measures on $\R^n$ with finite, positive total measure.
Suppose that $J_j\rightarrow J$ and $\omega_j\rightarrow\omega$ in the sense of \eqref{eqn:good regions_continuity}.
For each~$j$, let $\P_j$ be a Gibbs measures of the $(J_j,\omega_j)$-spin model, and suppose that $(\P_j)_{j\ge 1}$ converges in distribution to $\P$.
Then, $\P$ is a Gibbs measure of the $(J, \omega)$-spin model.
\end{proposition}

To prove Proposition \ref{prop:convergence-of-Gibbs-measures}, it is necessary to make the notion of prescribed boundary conditions, alluded to after \eqref{eq:box model finite-volume partition function periodic}, more explicit.
Given a finite, non-empty $\Lambda\subset\Z^d$ and a configuration $\tau:\Lambda^c\to\R^n$, the set of configurations with prescribed boundary conditions $\tau$ is
\begin{equation}
    \Omega^{\Lambda,\tau}:=\set{\eta:\Z^d\to\R^n\mid\eta_v=\tau_v\text{ for all }v\in\Lambda^c},
\end{equation}
the corresponding finite-volume Hamiltonian $H^{\Lambda,\tau}_J:\Omega^{\Lambda,\tau}\to\R$ is defined by
\begin{equation}
    H^{\Lambda,\tau}_J(\eta) := J\sum_{\substack{v\sim w\\\set{v,w}\cap\Lambda\ne\emptyset}}\norm{\eta_v-\eta_w}^2,
\end{equation}
and the corresponding finite-volume Gibbs measure is the probability measure $\P_{J,\omega}^{\Lambda,\tau}$ on $\Omega^{\Lambda,\tau}$ given by
\begin{equation}
    \P_{J,\omega}^{\Lambda,\tau}(\dd{\eta}) := \frac1{\Xi^{\Lambda,\tau}_{J,\omega}} e^{-H^{\Lambda,\tau}_J(\eta)}\prod_{v\in\Lambda} \omega(\dd{\eta}_v)
\end{equation}
where $\Xi^{\Lambda,\tau}_{J,\omega}$ is the normalization constant which makes $\P^{\Lambda,\tau}_{J,\omega}$ into a probability measure.

We start by making the following observation.
\begin{lemma}
\label{lem:convergence-of-Gibbs-measures_continuity}
    Let $\Lambda\subset\Z^d$ be finite and $f:\Omega\rightarrow\R$ be bounded and continuous.
    Then, $\E_{J,\omega}^{\Lambda,\tau}[f]$ is continuous respectively in $(J,\omega)$ and in $\tau$.
\end{lemma}

\begin{proof}
    Recall the normalization $\bar{\omega}$ of $\omega$ from \eqref{eqn:normalized probability measure}.
    Write
    \begin{align}
        \E_{J,\omega}^{\Lambda,\tau}[f]
        {}&=\frac{1}{\bar{\Xi}_{J,\omega}^{\Lambda,\tau}}\int\prod_{v\in\Lambda}\bar{\omega}(\dd{\eta}_v)e^{-H_{J}^{\Lambda,\tau}(\eta)}f(\eta),
        \label{eqn:convergence-of-Gibbs-measures_auxiliary_expectation}
        \\
        \bar{\Xi}_{J,\omega}^{\Lambda,\tau}{}&:=
        \int\prod_{v\in\Lambda}\bar{\omega}(\dd{\eta}_v)e^{-H_{J}^{\Lambda,\tau}(\eta)}.
    \end{align}
    The continuity of $\E_{J,\omega}^{\Lambda,\tau}[f]$ in $\tau$ follows from the bounded convergence theorem and the continuity of $H_{J}^{\Lambda,\tau}$ and $f$ in $\tau$.
    For the continuity of $\E_{J,\omega}^{\Lambda,\tau}[f]$ in $(J,\omega)$, we rely on the following elementary observation: given $0<a\le b$, there exists a constant $C>0$ such that for all $x,y\in[a,b]$ and $t\ge 0$, $\abs{e^{-tx}-e^{-ty}}\le C\abs{x-y}$.
    Let $((J_j,\omega_j))_{j\ge 1}$ be a sequence converging to $(J,\omega)$ as $j\rightarrow\infty$.
    We write
    \begin{equation}
    \label{eqn:convergence-of-Gibbs-measures_auxiliary_partition-function}
        \bar{\Xi}_{J_j,\omega_j}^{\Lambda,\tau}
        =\int\prod_{v\in\Lambda}\bar{\omega}_{j}(\dd{\eta}_v)\left[e^{-H_{J_j}^{\Lambda,\tau}(\eta)}-e^{-H_{J}^{\Lambda,\tau}(\eta)}\right]
        +\int\prod_{v\in\Lambda}\bar{\omega}_{j}(\dd{\eta}_v)e^{-H_{J}^{\Lambda,\tau}(\eta)}.
    \end{equation}
    We bound the first integral as follows.
    Recalling the form \eqref{Jw_ham} of the Hamiltonian and using that the sequence $(J_j)_{j\ge 1}$ is necessarily bounded, we deduce using the earlier observation that there exists a constant $C>0$ such that 
    \begin{equation}
        \abs{\exp{-H_{J_j}^{\Lambda,\tau}(\eta)}-\exp{-H_{J}^{\Lambda,\tau}(\eta)}}\le C\abs{J_{j}-J},
    \end{equation}
    so the first integral of \eqref{eqn:convergence-of-Gibbs-measures_auxiliary_partition-function} is bounded in modulus by $C\abs{J_j-J}$, which vanishes as $j\rightarrow\infty$.
    For the second integral of \eqref{eqn:convergence-of-Gibbs-measures_auxiliary_partition-function}, we note that $\bar{\omega}_{j}\rightarrow\bar{\omega}$ in distribution implies the convergence in distribution of the corresponding product measures \cite[Theorem 2.8]{billingsley2013convergence}: $\prod_{v\in\Lambda}\bar{\omega}_{j}\rightarrow\prod_{v\in\Lambda}\bar{\omega}$.
    By the continuity and non-negativity of the Hamiltonian, we conclude that the second integral converges to $\bar{\Xi}_{J,\omega}^{\Lambda,\tau}$ as $j\rightarrow\infty$.
    Having thus shown the continuity of $\bar{\Xi}_{J,\omega}^{\Lambda,\tau}$ in $(J,\omega)$, we note that the same argument applies to the integral in \eqref{eqn:convergence-of-Gibbs-measures_auxiliary_expectation}, which completes the proof.
\end{proof}

We now deduce Proposition \ref{prop:convergence-of-Gibbs-measures}.

\begin{proof}[Proof of Proposition \ref{prop:convergence-of-Gibbs-measures}]
Our goal is to prove that $\P$ verifies the DLR condition \cite[Definition 2.9]{georgii2011gibbs} with the Gibbsian specifications of the $(J,\omega)$-spin model.
By \cite[Chapter 3, Proposition 4.6(b)]{ethier2009markov}, it suffices to show that, for all finite $\Lambda\subset\Z^d$ and bounded, continuous $f:\Omega\rightarrow\R$,
\begin{equation}
    \E[f\mid\eta_{\Lambda^c}]=\E^{\Lambda,\eta_{\Lambda^c}}_{J,\omega}[f],
\end{equation}
which is, in turn, verified if for all bounded, continuous, and $\mathscr{F}_{\Lambda^c}$-measurable $g:\Omega\rightarrow\R$,
\begin{equation}
\label{eqn:convergence-of-Gibbs-measures_goal}
    \E[g\cdot\E[f\mid\eta_{\Lambda^c}]]
    =\E[g\cdot\E^{\Lambda,\eta_{\Lambda^c}}_{J,\omega}[f]].
\end{equation}

As $g$ is $\mathscr{F}_{\Lambda^c}$-measurable, the LHS of \eqref{eqn:convergence-of-Gibbs-measures_goal} reduces to $\E[f\cdot g]$.
In the meantime, we write the RHS of \eqref{eqn:convergence-of-Gibbs-measures_goal} as
\begin{equation}
\label{eqn:convergence-of-Gibbs-measures_rhs}
\begin{multlined}
    \E[g\cdot\E^{\Lambda,\eta_{\Lambda^c}}_{J,\omega}[f]]
    =\left(\E[g\cdot\E^{\Lambda,\eta_{\Lambda^c}}_{J,\omega}[f]]
    -\E_{j}[g\cdot\E^{\Lambda,\eta_{\Lambda^c}}_{J,\omega}[f]]\right)
    \\
    +\E_{j}[g\cdot(\E^{\Lambda,\eta_{\Lambda^c}}_{J,\omega}[f]-\E^{\Lambda,\eta_{\Lambda^c}}_{J_j,\omega_j}[f])]
    +\E_{j}[g\cdot\E^{\Lambda,\eta_{\Lambda^c}}_{J_j,\omega_j}[f]],
\end{multlined}
\end{equation} 
where the last term further reduces to
\begin{equation}
    \E_{j}[g\cdot\E^{\Lambda,\eta_{\Lambda^c}}_{J_j,\omega_j}[f]]
    =\E_{j}[g\cdot\E_{j}[f\mid\eta_{\Lambda^c}]]
    =\E_{j}[f\cdot g].
\end{equation}
By Lemma \ref{lem:convergence-of-Gibbs-measures_continuity}, the first two terms on the RHS of \eqref{eqn:convergence-of-Gibbs-measures_rhs} both vanish as $j\rightarrow\infty$, so
\begin{equation}
    \E[g\cdot\E[f\mid\eta_{\Lambda^c}]]
    =\E[f\cdot g]
    =\lim_{j\rightarrow\infty}\E_{j}[f\cdot g]
    =\E[g\cdot\E^{\Lambda,\eta_{\Lambda^c}}_{J,\omega}[f]],
\end{equation}
as required.
\end{proof}

\end{document}